\documentclass[letterpaper, 11pt]{article}
\usepackage{amssymb}
\usepackage{amsmath}
\usepackage{natbib}
\usepackage{color}
\usepackage{listings}
\usepackage{graphicx,enumerate}
\usepackage[space]{grffile}
\usepackage{subfigure}
\usepackage{mathtools}
\usepackage{caption}
\usepackage{graphicx,xcolor,colortbl}    
\usepackage{epstopdf,epsfig} 
\usepackage{hyperref}
\hypersetup{ 
  colorlinks=false,
  pdfborder = {0 0 0.5 [3 0]}
}

\newtheorem{theorem}{Theorem}

\newtheorem{proposition}[theorem]{Proposition}
\newtheorem{remark}[theorem]{Remark}

\numberwithin{equation}{section} 
\newenvironment{proof}[1][Proof]{\textbf{#1.} }{\ \rule{0.5em}{0.5em}}
\def\Q{{\mathbb Q}}        
\def\R{{\mathbb R}}        
\def\P{{\mathbb P}}        
\def\E{{\mathbb E}}        
 
\def\1{{\mathbf 1}}        

\def\L{{\mathcal L} \,}

\def\setT{{\mathcal T}}

\def\Vk1{{V^{(i-1)}}}

\def\Hk1{{H^{(i-1)}}}

\def\pk1{{p^{(i-1)}}}

\def\vs{{\varsigma}}

\def\VV{{\mathcal V}}
\def\JJ{{\mathcal J}}
\def\UU{{\mathcal U}}
\def\KK{{\mathcal K}}
\def\PP{{\mathcal P}}
\def\AA{{\mathcal A}}
\def\BB{{\mathcal B}}

\usepackage{multirow}
\usepackage[left=1in,right=1in,top=1.1in,bottom=1.1in]{geometry}

\begin{document}
\title{Speculative Futures Trading under Mean Reversion\thanks{The authors would like to thank Sebastian Jaimungal and Peng Liu for their helpful remarks, as well as the  participants of the Columbia-JAFEE Conference 2015, especially Jiro Akahori, Junichi Imai,  Yuri Imamura, Hiroshi Ishijima, Keita Owari, Yuji Yamada, Ciamac Moallemi,   Marcel Nutz, and Philip Protter. }}
\author{Tim Leung\thanks{\mbox{Corresponding author.}} \thanks{IEOR Department, Columbia University, New York, NY 10027; email:\,\mbox{leung@ieor.columbia.edu}.} \and Jiao Li\thanks{APAM Department, Columbia University, New York, NY 10027; email:\,\mbox{jl4170@columbia.edu}.}  \and Xin Li\thanks{IEOR Department, Columbia University, New York, NY 10027; email:\,\mbox{xl2206@columbia.edu}.}  \and Zheng Wang\thanks{IEOR Department, Columbia University, New York, NY 10027; email:\,\mbox{zw2192@columbia.edu}.}}
\date{November 25, 2015}
\maketitle
\begin{abstract}
This paper studies the problem of trading futures with transaction costs when the underlying spot  price  is mean-reverting. Specifically, we  model the spot dynamics by the Ornstein-Uhlenbeck (OU),  Cox-Ingersoll-Ross (CIR), or  exponential Ornstein-Uhlenbeck (XOU) model.  The futures term structure is derived and its connection to futures price dynamics is examined.   For  each futures contract, we describe the  evolution of the roll yield, and  compute explicitly the expected roll yield.  For the futures trading problem, we incorporate     the investor's  timing option to enter or exit the market, as well as a chooser option to long or short a futures upon entry. This leads us to formulate and  solve the corresponding   optimal double stopping problems to determine the optimal trading strategies. Numerical results are presented to illustrate the optimal entry and exit boundaries under different models. We find  that the option to choose between a long or short position induces the investor to  delay market entry, as compared to the case where the investor pre-commits to go either long or short.
\end{abstract}
\vspace{10pt}

 \newpage
\section{Introduction\label{sect-intro}}
Futures are an integral part of the universe of derivatives. In 2014, the  total number of futures and options contracts traded on exchanges worldwide rose 1.5\% to 21.87 billion from 21.55 billion in 2013, with futures contracts alone accounting for 12.17 billion of these contracts. The CME group and Intercontinental Exchange are the two largest futures and options exchanges. The 2014 combined trading volume of  CME group with its subsidiary exchanges, Chicago Mercantile Exchange, Chicago Board of Trade and  New York Mercantile Exchange  was  3.44 billion  contracts,  while Intercontinental Exchange had a   volume of 2.28 billion contracts.\footnote{Statistics taken from \cite{acworth2015}.  }

A futures is a contract that requires the  buyer to purchase (seller to sell) a fixed quantity of an asset, such as a commodity, at a fixed price to be paid for on a pre-specified future date. Commonly traded on exchanges, there are futures written  on various underlying assets or references, including  commodities, interest rates,  equity indices, and volatility indices. Many futures stipulate physical delivery of the underlying asset, with notable examples of	 agricultural, energy, and metal futures. However, some, like the VIX futures, are settled in cash.

Futures are often  used as a hedging instrument, but they are also  popular among speculative investors. In fact, they are seldom traded with the intention of holding it to maturity as less than 1\% of futures traded ever reach physical delivery.\footnote{See  p.615 of \cite{mpt2009} for a discussion.}  This motivates  the question of optimal timing to trade a futures.

In this paper, we investigate  the speculative trading of futures under mean-reverting spot price dynamics.  Mean reversion is commonly observed for the spot price in many futures markets,  ranging from commodities and interest rates  to currencies and volatility indices, as studied in many  empirical studies (see, among others,  \cite{bessembinder1995}, \cite{scottfuturesmeanreversion1996}, \cite{schwartz1997stochastic}, \cite{Casassus2005}, \cite{gemanoilmeanreversion2007}, \cite{Balimeanreversion2008},    \mbox{\cite{wangVIX2011}}). For volatility  futures as an example,   \cite{Grunbichler1996985}  and  \cite{futures_zhang} model the S\&P500 volatility index (VIX) by  the  Cox-Ingersoll-Ross (CIR) process and provide a formula for the futures price.   We start by  deriving the price functions and  dynamics of the futures under the Ornstein-Uhlenbeck (OU), CIR, and exponential OU (XOU) models. Futures prices are computed under the risk-neutral measure, but its evolution over time is described by the historical measure. Thus, the investor's optimal timing to trade depends on both measures.

Moreover, we incorporate the investor's timing option to enter and subsequently exit the market. Before entering the market, the investor faces two possible strategies: long or short a futures first, then  close the position later.  In the first strategy, an investor is expected to establish the long position when the price is sufficiently  low, and then  exits when the price is high. The opposite is expected  for the second strategy. In both cases, the presence of transaction costs expands the waiting region, indicating the investor's desire for better prices. In addition, the waiting region expands drastically near expiry since transaction costs discourage entry when futures is very close to maturity. Finally, the main feature of our trading problem approach is to combine these two related problems and analyze the optimal strategy when an investor has the freedom to choose between either a \emph{long-short} or a \emph{short-long} position. Among our results, we find that when the investor has the right to choose, she delays market entry to wait for better prices compared to the individual standalone problems.

Our model is a variation of the theoretical arbitrage model proposed by \cite{futuresDaiKwok}, who also  incorporate the timing options to enter and exit the market, as well as the choice between opposite positions upon entry. Their sole underlying traded process is the stochastic \emph{basis} representing the difference between the index and futures values, which is modeled by a Brownian bridge. In an earlier study,  \cite{futuresBS}   formulate a similar optimal stopping problem for trading futures where the underlying basis is a Brownian bridge.  In comparison to these two models,  we model directly the spot price process, which allows for  calibration of futures prices and provide a no-arbitrage link between the (risk-neutral) pricing  and (historical) trading problems, as opposed to a priori assuming the existence of arbitrage opportunities, and  modeling the basis that is neither  calibrated nor shown to be consistent with the futures curves.  A similar timing strategy  for pairs trading has been studied by \cite{HFTbook} as an extension of the buy-low-sell-high strategy used in \cite{LeungLi2014OU}.

In addition, we  study the distribution and dynamics of \emph{roll yield},  an important concept in futures trading.  Following the literature and industry practice, we  define roll yield as the difference between changes in futures price and changes in the underlying price (see e.g.  \cite{moskowitz2012time}, \cite{gorton2013fundamentals}). For traders,  roll yield is a useful gauge for deciding to invest  in   the spot asset or associated  futures. In essence, roll yield defined herein represents the net cost and/or benefit of owning futures over the spot asset. Therefore, even  for an investor who trades futures only, the corresponding roll yield is a useful reference and can affect  her trading decisions. 

The rest of the paper is structured as follows.   Section \ref{sect-futuresprices} summarizes the   futures prices and term structures under mean reversion.  We discuss the concept of roll yield   in Section \ref{sect-rollyield}. In Section \ref{sect-futtrading}, we formulate and numerically solve the  optimal double stopping problems for   futures trading.  Our numerical algorithm is described in the Appendix.

\section{Futures Prices and Term Structures}\label{sect-futuresprices}
Throughout this paper, we consider  futures that are written on an asset whose price process is mean-reverting. In this section, we discuss the pricing of futures\index{futures}  and their term structures under different spot models.

 \subsection{OU and CIR Spot Models}\label{sect-OUCIRspot}
 We begin with two mean-reverting models for the spot price $S$, namely, the OU and CIR models.  As we will see, they yield the same price function for the futures contract. To start, suppose that   the spot price evolves according to the OU model:
\begin{align*}
dS_t = {\mu}( {\theta}-S_t)dt+\sigma dB_t,
\end{align*}
where ${\mu},\sigma>0$ are the speed of mean reversion and volatility of the process respectively. ${\theta}\in\R$ is the long run mean and $B$ is a standard Brownian motion under the historical measure $\P$. 

To price futures, we assume a re-parametrized OU model for  the risk-neutral spot price dynamics. Hence, under the risk-neutral measure $\Q$, the spot price follows
\begin{align*}
dS_{t}=  \tilde{\mu}(\tilde{\theta} - S_t)\,dt+\sigma \,dB_{t}^{\Q},
\end{align*} with  constant parameters $\tilde{\mu}, \sigma>0$, and $\tilde{\theta} \in \R$. This is again an OU process, albeit with a different long-run mean $\tilde{\theta}$ and speed of mean reversion $\tilde{\mu}$ under the risk-neutral measure. This involves a  change of measure that connects the two Brownian motions, as described by 
\begin{align*}
dB_{t}^{\Q} = dB_{t} +  \frac{\mu(\theta - S_t) - \tilde{\mu}(\tilde{\theta} - S_t)}{\sigma}dt .
\end{align*}

Throughout,  futures  prices are computed the same as  forward prices, and  we do not distinguish between the two prices (see  \citet{CIR1981,futuresBS}).  As such, the price of a futures contract    with maturity $T$ is given by
\begin{align}\label{fTOU}
f^T_t \equiv f(t, S_t;T) := \E^{\Q}\{S_T|S_t\} = (S_t-\tilde{\theta})e^{-\tilde{\mu}(T-t)}+\tilde{\theta}, \quad t\le T. 
\end{align}  Note that the futures price is a deterministic function of time and the current spot price.

We now  consider the CIR model for  the spot price:
\begin{align}\label{VIXCIRP}
dS_t ={\mu}({\theta}-S_t)dt+\sigma\sqrt{S_t} d{B}_t,
\end{align}
where ${\mu}, {\theta},\sigma>0$, and ${B}$ is a standard Brownian motion under the historical measure $\P$. Under the risk-neutral measure $\Q$, 
\begin{align}\label{CIRQ}
dS_t = \tilde{\mu}(\tilde{\theta}-S_t)dt + \sigma\sqrt{S_t} dB_t^{\Q},
\end{align}
where $\tilde{\mu}, \tilde{\theta} >0$, and $B^{\Q}$ is a  $\Q$-standard Brownian motion. In both SDEs,  \eqref{VIXCIRP}  and \eqref{CIRQ}, we require   
$2\mu \theta \ge \sigma^2$ and $2\tilde{\mu}\tilde{\theta} \ge \sigma^2$ (Feller condition) so that the CIR process stays positive. 

The two Brownian motions are related by 
\begin{align*}
dB_t^{\Q}=d{B}_t+\frac{{\mu}({\theta}-S_t)-\tilde{\mu}(\tilde{\theta}-S_t)}{\sigma\sqrt{S_t}}  \,dt,
\end{align*}
  which preserves the CIR model, up to different parameter values across two measures.  

The CIR terminal spot price $S_T$ admits the non-central Chi-squared distribution and is positive, whereas the OU spot price is normally distributed.  Nevertheless, the futures price under the CIR model  admits the same functional form as in the OU case (see \eqref{fTOU}):
\begin{align}\label{fTCIR}
f^T_t  = (S_t-\tilde{\theta})e^{-\tilde{\mu}(T-t)} +\tilde{\theta}, \quad t\le T.
\end{align}
 
\begin{proposition}
Under the OU or CIR spot model, the futures curve is (i) upward-sloping and concave if the  current spot price $S_0<\tilde{\theta}$, (ii) downward-slopping and convex if $S_0 > \tilde{\theta}$. 
\end{proposition}
\begin{proof} We differentiate  \eqref{fTCIR} with respect to  $T$ to get the derivatives:
\begin{align*}
\frac{\partial f^T_0}{\partial T} = -\tilde{\mu} (S_0-\tilde{\theta})e^{-\tilde{\mu}T} \lessgtr 0 \quad \text{ and }  \quad 
\frac{\partial^2 f^T_0}{\partial T^2} = \tilde{\mu}^2(S_0-\tilde{\theta})e^{-\tilde{\mu}T}  \gtrless 0,
\end{align*}   for $S_0 \gtrless \tilde{\theta}.$ Hence, we conclude.
\end{proof}

\begin{remark} The futures price formula \eqref{fTCIR} holds more generally for other mean-reverting models with risk-neutral spot dynamics of the form: 
 \begin{align*}
dS_t = \tilde{\mu}(\tilde{\theta}-S_t)dt + \sigma(S_t) dB_t^{\Q},
\end{align*}where $\sigma(\cdot)$ is a deterministic function such that $\E^\Q\{\int_0^T\sigma(S_t)^2dt\}<\infty$.    
\end{remark}

 Under the OU model, the futures satisfies the  following  SDE under the historical measure $\P$:
\begin{align}\label{FutOU}
df_t^T = \left[(f_t^T- \tilde{\theta})(\tilde{\mu} - \mu) +  \mu(\theta - \tilde{\theta})e^{-\tilde{\mu}(T-t)}\right]dt + \sigma e^{-\tilde{\mu}(T-t)} dB_t.
\end{align}
If the spot follows a CIR process, then the futures prices follows
\begin{align}\label{FutCIR}
df_t^T = \left[(f_t^T- \tilde{\theta})(\tilde{\mu} - \mu) +  \mu(\theta - \tilde{\theta})e^{-\tilde{\mu}(T-t)}\right]dt + \sigma e^{-\tilde{\mu}(T-t)}\sqrt{(f^T_t - \tilde{\theta})e^{\tilde{\mu}(T-t)} +\tilde{\theta}}dB_t.
\end{align}
Notice that the same drift appears in both \eqref{FutOU} and \eqref{FutCIR}.  Alternatively, we can express the drift in terms of the spot price as
\begin{align*}
e^{-\tilde{\mu}(T-t)}(\mu(\theta - S_t) - \tilde{\mu}(\tilde{\theta} - S_t)).
\end{align*}
This involves the difference between the mean-reverting drifts of the spot price   under the historical measure $\P$ and the risk-neutral measure $\Q$.  Therefore, the  drift of the futures price SDE is positive when the drift of the spot price under $\P$ is greater than that under $\Q$, i.e.
\begin{align*}
\mu(\theta - S_t) > \tilde{\mu}(\tilde{\theta} - S_t),
\end{align*}
and vice versa.

Now, consider an investor with a long position in a single futures contract, she wishes to close out the position and is interested in determining the best time to short. We consider the \emph{delayed liquidation premium}, which was introduced in \cite{LeungShirai} for equity options.  This premium  expresses the benefit of waiting to liquidate as compared to closing the position  immediately. Precisely, the delayed liquidation premium  is defined as 
\begin{align}\label{def_premium}
L(t,s) := \sup_{\tau \in\setT_{t,T}}{\E}_{t,s}\!\big\{e^{-r (\tau-t)}(f(\tau, S_{\tau};T)-c) \big\} - (f(t, s;T)-c),
\end{align}
where $\setT_{t,T}$  is the set of all stopping times, with respect to the filtration generated by $S$,  and $c$ is the transaction cost. As we can see in \eqref{def_premium}, the optimal stopping time for $L(t,s)$, denoted by $\tau^*$,   maximizes the expected discounted   value from liquidating the futures. 

\begin{proposition}\label{prop_ou/cir_premium}
Let $t \in [0, T]$ be the current time, and define the function 
\begin{align*}
G(u,s) := e^{-\tilde{\mu}(u-t)}(\mu(\theta - s) + (r - \tilde{\mu})(\tilde{\theta} - s))   + r(c - \tilde{\theta}).
\end{align*}Under the OU spot model, if $G(u, s)\ge 0$,
$\forall (u, s) \in [t, T] \times \R$, then it is optimal to hold the futures contract till expiry, namely, $\tau^* = T$ in \eqref{def_premium}. If
$G(u, s)< 0$, $\forall (u, s) \in [t, T] \times \R$, then it is optimal to liquidate immediately, namely, $\tau^* = t.$ The same holds under the CIR model with $G(u,s)$ defined over $ [t, T] \times \R_+$. 
\end{proposition}
\begin{proof}Applying Ito's formula to the   process of $e^{-rt}(f_t^T - c)$ and taking expectation, we can express  \eqref{def_premium} as 
\begin{align}\label{ou/cir_premium}
L(t,s) = \sup_{\tau \in\setT_{t,T}}{\E}_{t,s}\left\{\int_t^{\tau} e^{-r (u-t)}\left[ e^{-\tilde{\mu}(u-t)}(\mu(\theta - S_u) + (r - \tilde{\mu})(\tilde{\theta} - S_u))   + r(c - \tilde{\theta}) \right]du \right\}.
\end{align}
Therefore, if $G(u,s)$ (the integrand in \eqref{ou/cir_premium}) is positive, $\forall (u, s) \in [t, T] \times \R$, then the delayed liquidation premium can be maximized by choosing $\tau^* = T,$ which is the largest stopping time. Conversely, if $G<0$  $\forall (u, s) \in [t, T]\times \R$, then it is optimal to take $\tau^* = t$ in \eqref{ou/cir_premium}. Note that if $G=0$  $\forall (u, s) \in [t, T]\times\R$, then the delayed  liquidation premium is zero, and the investor is indifferent toward when to liqudiate.  \end{proof}

\subsection{Exponential OU Spot Model}\label{sect-XOUspot}
Under the exponential OU (XOU)  model, the spot price follows the SDE:
\begin{align}\label{XOUP}
dS_t = {\mu}({\theta}-\ln(S_t))S_tdt + \sigma S_t d{B}_t,
\end{align}
with positive parameters $(\mu, \theta, \sigma)$, and  standard Brownian motion $B$ under  the historical measure $\P$. For pricing futures, we assume that  the risk-neutral dynamics of $S$ satisfies 
\begin{align*}
dS_t =\tilde{\mu}(\tilde{\theta}-\ln(S_t))S_tdt + \sigma S_t dB_t^{\Q},
\end{align*}
where   $\tilde{\mu}, \tilde{\theta} >0$, and $B^{\Q}$  is a standard Brownian motion under the risk-neutral measure $\Q$.

For a  futures contract written on $S$ with maturity $T$, its price at time $t$ is given by
\begin{align}\label{fTXOU}
f^T_t  & = \exp\bigg(e^{-\tilde{\mu}(T-t)}\ln(S_t) + (1-e^{-\tilde{\mu}(T-t)})(\tilde{\theta}-\frac{\sigma^2}{2\tilde{\mu}})
 + \frac{\sigma^2}{4\tilde{\mu}}(1- e^{-2\tilde{\mu}(T-t)} )\bigg).
\end{align}
Consequently, the dynamics of the futures price under the historical measure $\P$ is given as
\begin{align}
df^{T}_t &= \left[ \left(\ln(f^T_t) + (e^{-\tilde{\mu}(T-t)}-1)(\tilde{\theta}-\frac{\sigma^2}{2\tilde{\mu}})  + \frac{\sigma^2}{4\tilde{\mu}}(e^{-2\tilde{\mu}(T-t)}-1) \right)(\tilde{\mu} - \mu)  \right. \notag\\ 
  &\quad \left. {}
+ e^{-\tilde{\mu}(T-t)}(\mu\theta - \tilde{\mu}\tilde{\theta})\right]f^{T}_tdt  + \sigma e^{-\tilde{\mu}(T-t)}f^{T}_tdB_t.\label{xou_fut_sde}
\end{align}
By rearranging the first term  in \eqref{xou_fut_sde}, the  drift of the futures price SDE  is positive iff
\begin{align*}
f_t^T &> \exp\left[\frac{e^{-\tilde{\mu}(T-t)}(\tilde{\mu}\tilde{\theta} - \mu\theta) }{\tilde{\mu} - \mu} - (e^{-\tilde{\mu}(T-t)}-1)(\tilde{\theta}-\frac{\sigma^2}{2\tilde{\mu}})  - \frac{\sigma^2}{4\tilde{\mu}}(e^{-2\tilde{\mu}(T-t)}-1)\right],
\end{align*}
or equivalently in terms of the spot price, 
\begin{align}
S_t > \exp\left(\frac{\tilde{\mu}\tilde{\theta} - \mu\theta}{\tilde{\mu} - \mu}\right). \label{Sgreater}
\end{align}
In particular, if  $\tilde{\theta} = \theta$, condition \eqref{Sgreater} reduces to $\log S_t > \theta$. Intuitively, since the  futures price must converge to the spot price at maturity,  the futures price tends to rise to approach the spot price when the spot price is  high, as observed in this condition.

We now consider  the delayed liquidation premium defined in \eqref{def_premium} but under the XOU spot model. Applying Ito's formula, we express the  optimal liquidation premium as
\begin{align}
L(t,s) = \sup_{\tau \in\setT_{t,T}}{\E}_{t,s}\left\{\int_t^{\tau} e^{-r (u-t)}\widetilde{G}(u, S_u)du  \right\}, \label{XOULts}
\end{align} 
where
\begin{align}
\widetilde{G}(u,s) := & \left\{r + \left[{\mu}({\theta}-\ln(s)) -  \tilde{\mu}(\tilde{\theta}-\ln(s))\right]e^{-\tilde{\mu}(u-t)}\right\}\notag\\
\times & \exp\bigg(e^{-\tilde{\mu}(u-t)}\ln(s) + (1-e^{-\tilde{\mu}(u-t)})(\tilde{\theta}-\frac{\sigma^2}{2\tilde{\mu}})
 + \frac{\sigma^2}{4\tilde{\mu}}(1- e^{-2\tilde{\mu}(u-t)} )\bigg)  - rc. \label{GXOU}
\end{align}
By inspecting the premium definition, we obtain the condition under which immediate liquidation or waiting till maturity is optimal. The proof is identical to that of Proposition \ref{prop_ou/cir_premium}, so we omit it.
\begin{proposition}
Let $t \in [0, T]$ be the current time.
Under the XOU spot model, if   $\widetilde{G}(u, s)\ge 0$  $\forall (u, s) \in [t, T] \times \R_+$, then holding till maturity ($\tau^* = T$) is optimal for  \eqref{XOULts}. If
$\widetilde{G}(u, s)<0$, $\forall (u, s) \in [t, T] \times \R_+$, then immediate liquidation ($\tau^* = t$) is optimal for \eqref{XOULts}.
 \end{proposition}

Next, we summarize the term structure of futures under the XOU spot model. 
\begin{proposition}
Under the XOU spot model, the futures curve is

\begin{enumerate}[(i)]
\item downward-sloping and convex if 
\[\ln S_0 >  \tilde{\theta} - \frac{\sigma ^2}{2 \tilde{\mu}}(1 - e^{-\tilde{\mu}T}) + \left(\frac{e^{2 \tilde{\mu}T}}{4} + \frac{\sigma ^2}{2 \tilde{\mu}}\right)^{\frac{1}{2}} - \frac{e^{\tilde{\mu}T}}{2},\]
\item downward-sloping and concave if 
\[  \tilde{\theta} - \frac{\sigma ^2}{2 \tilde{\mu}}(1 - e^{-\tilde{\mu}T}) < \ln S_0 < \tilde{\theta} - \frac{\sigma ^2}{2 \tilde{\mu}}(1 - e^{-\tilde{\mu}T}) + \left(\frac{e^{2 \tilde{\mu}T}}{4} + \frac{\sigma ^2}{2 \tilde{\mu}}\right)^{\frac{1}{2}} - \frac{e^{\tilde{\mu}T}}{2} ,\]
\item upward-sloping and concave if 
\[ \tilde{\theta} - \frac{\sigma ^2}{2 \tilde{\mu}}(1 - e^{-\tilde{\mu}T}) -\left(\frac{e^{2 \tilde{\mu}T}}{4} + \frac{\sigma ^2}{2 \tilde{\mu}}\right)^{\frac{1}{2}} - \frac{e^{\tilde{\mu}T}}{2}  <\ln  S_0 <   \tilde{\theta} - \frac{\sigma ^2}{2 \tilde{\mu}}(1 - e^{-\tilde{\mu}T}),\] 

and

\item upward-sloping and convex if \[  \ln S_0 <  \tilde{\theta} - \frac{\sigma ^2}{2 \tilde{\mu}}(1 - e^{-\tilde{\mu}T}) - \left(\frac{e^{2 \tilde{\mu}T}}{4} + \frac{\sigma ^2}{2 \tilde{\mu}}\right)^{\frac{1}{2}} - \frac{e^{\tilde{\mu}T}}{2}.\]
\end{enumerate}

\end{proposition}
\begin{proof}
Direct differentiation of $f_0^T$ yields that
\begin{align*}
\frac{\partial f^T_0}{\partial T} = \left[ \tilde{\mu} (\tilde{\theta} - \frac{\sigma ^2}{2 \tilde{\mu}} - \ln S_0 ) e^{-\tilde{\mu}T} + \frac{\sigma ^2}{2} e^{-2 \tilde{\mu}T} \right] f_0 ^T,
\end{align*}
and
\begin{align*}
\begin{split}
\frac{\partial^2 f^T_0}{\partial T^2} &= \bigg[\tilde{\mu} ^2 e^{-2 \tilde{\mu}T} ( \tilde{\theta} -\frac{\sigma ^2}{2 \tilde{\mu}} - \ln S_0 )^2 + (\tilde{\mu} \sigma ^2  e^{-3 \tilde{\mu}T} - \tilde{\mu}^2  e^{- \tilde{\mu}T}) ( \tilde{\theta} - \frac{\sigma ^2}{2 \tilde{\mu}} -\ln S_0)\\
 &\quad +\frac{\sigma ^4}{4} e^{-4 \tilde{\mu}T} -\sigma ^2 \tilde{\mu}  e^{-2 \tilde{\mu}T}\bigg] f_0 ^T.
 \end{split}
\end{align*}
The results are obtained by analyzing the signs of the first and second order derivatives.
\end{proof}

\newpage

 \begin{figure}[h]
   \centering
\includegraphics[trim=5   1  7  1,clip,width=3in]{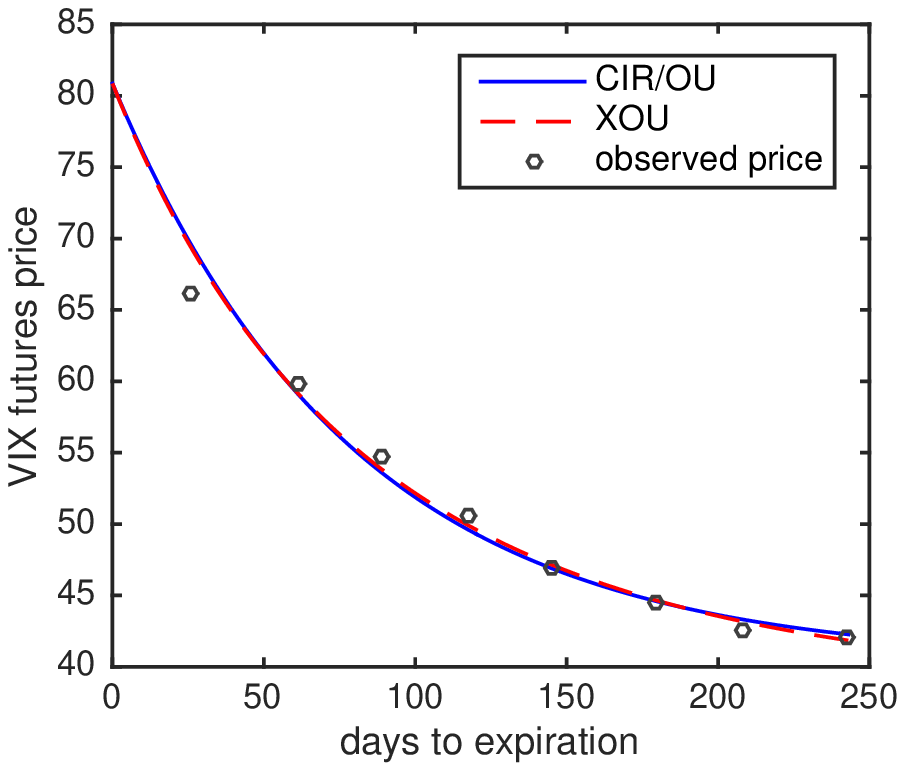}
\includegraphics[trim=5  1 7  1,clip,width=3in]{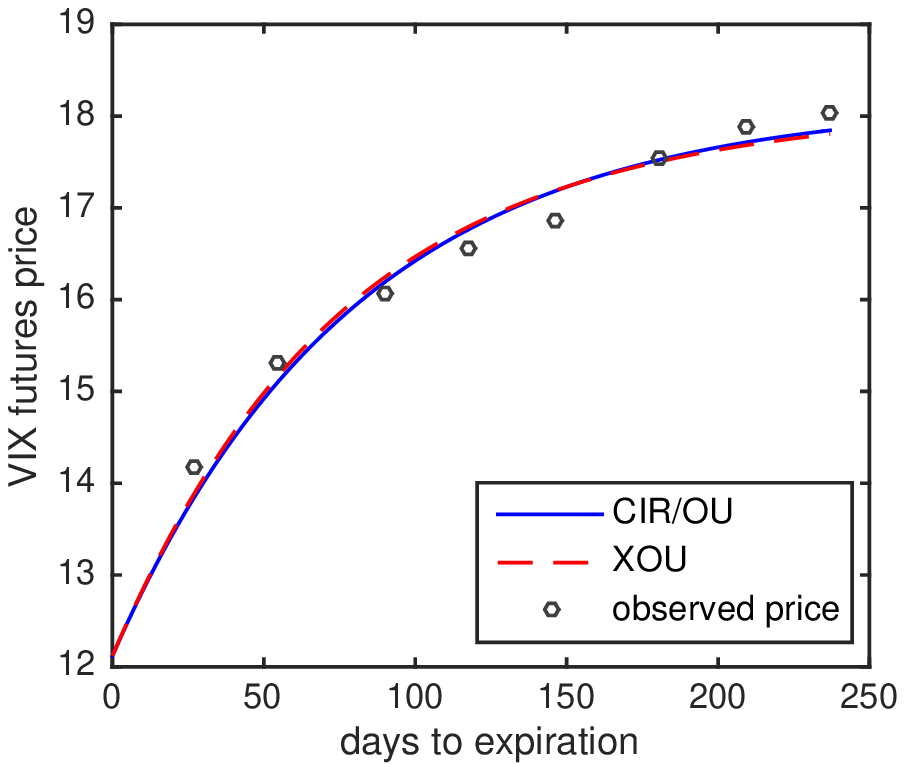}
   \caption{ (Left) VIX futures historical prices on Nov 20, 2008 with the current VIX value at  80.86. The days to expiration range from 26   to 243 days (Dec--Jul contracts). Calibrated parameters:  $\tilde{\mu} = 4.59, \tilde{\theta} = 40.36$ under the  CIR/OU model, or  $\tilde{\mu} = 3.25, \tilde{\theta} = 3.65, \sigma = 0.15$ under the XOU model. (Right) VIX futures historical prices on Jul 22, 2015 with the current VIX value at 12.12. The days to expiration ranges from 27 days to 237 days (Aug--Mar contracts). Calibrated parameters:  $\tilde{\mu} = 4.55, \tilde{\theta} = 18.16$ under the CIR/OU model, or   $\tilde{\mu} = 4.08, \tilde{\theta} = 3.06, \sigma = 1.63$ under the XOU model.\label{fig-futuresterm}}
   \end{figure}

 Figure \ref{fig-futuresterm} displays two characteristically different  term structures observed in the VIX futures market.   These  futures, written  on the  CBOE Volatility Index\index{volatility index} (VIX)  are traded on  the CBOE Futures Exchange.  As the VIX measures the  1-month implied volatility calculated  from the prices of S\&P 500 options, VIX futures  provide exposure to the market's volatility. We plot  the   VIX futures prices during the recent financial crisis on November 20, 2008 (left),  and on a post-crisis date, July 22, 2015 (right),  along with the calibrated  futures curves under the OU/CIR model and XOU model.  In the calibration, the model parameter values  are chosen to minimize the sum of squared errors between the model   and observed futures prices.  
 
 The OU/CIR/XOU model  generates a decreasing convex curve for November 20, 2008 (left), and an  increasing concave curve for July 22, 2015 (right), and they all fit the observed futures prices very well. The former  term structure  starts with a very high spot price of 80.86 with a calibrated risk-neutral long-run mean $\tilde{\theta}=40.36$ under the OU/CIR model, suggesting that the market's expectation of falling market volatility. In contrast, we infer from  the  term structure  on July 25, 2015  that the market expects the VIX to raise from the current spot value of 12.12 to be closer to $\tilde{\theta} = 18.16$.

\section{Roll Yield}\label{sect-rollyield}

By design,  the value of a futures contract converges to the spot price as time approaches maturity. If the futures market is in  \emph{backwardation},  the futures  price  increases to reach  the spot price at expiry. In contrast,     when the market is in contango\index{contango}, the futures   price tends to decrease  to the spot price. For an  investor  with a long futures position, the return is positive in a backwardation market, and negative in a contango market. An investor can long the front-month contract, then short it at or before expiry, and  simultaneously go long the next-month contract.   This   \emph{rolling strategy} that involves repeatedly rolling an expiring contract into a new one     is commonly adopted during backwardation, while its opposite is often used in  a contango market. Backwardation and contango phenomena are widely observed in the energy commodities and volatility futures markets.

More generally, both the futures and spot prices vary over time.  If the spot price increases/decreases, the futures  price will also end up higher/lower.  This leads us to consider the difference between the futures and spot returns, defined as the  change in   values without dividing by the initial value.\footnote{See \textit{Deconstructing Futures Returns: The Role of Roll Yield}, Campbell White Paper Series, February 2014.} Let $0\leq t_1 < t_2 \leq T$. We  denote the \textit{roll yield}\index{roll yield} over the period $[t_1, t_2]$ associated with a single futures contract with maturity $T$ by
\begin{align} \label{rollyielddef}
\mathcal{R}(t_1,t_2,T) := (f^T_{t_2}-f^T_{t_1})-(S_{t_2}-S_{t_1}).
\end{align}
In other words, roll yield here is the change in the futures price that is not accounted for by the change in spot price. It represents the  net benefits and/or costs of owning futures rather than the underlying asset itself. 

This notion of roll yield is the same as that in \cite{moskowitz2012time}  where the relationship between roll yield and futures returns is studied. \cite{gorton2013fundamentals} treat roll yield as the same as futures basis, which means  a negative roll yield signifies a market in contango and a positive roll yield is equivalent to backwardation. In our set-up, if one always hold a futures contract to maturity, then roll yield is the same as futures basis. Therefore, the definition of roll yield in \cite{gorton2013fundamentals} is a special case of ours. In particular, if $t_2=T$, then the roll yield reduces to the price difference $(S_{t_1}-f^T_{t_1})$. Furthermore, observe that if $S_{t_2} = S_{t_1}$ then roll yield becomes merely the change in futures price. 

A closely related concept is the S\&P-GSCI roll yield. S\&P-GSCI 
carries out rolling of the underlying futures contracts once each month, from the fifth to the ninth business day. On each day, 20\% of the current portfolio is rolled over, in a process commonly  known  as the \emph{Goldman roll}. The S\&P-GSCI roll yield for each commodity is   defined as the difference between the average purchasing price of the new futures contracts and the average selling price of the old futures contracts. In essence,  it is an indicator of the sign of the slope of the futures term structure. In comparison to the S\&P-GSCI index, our definition accounts for the changes of spot price over time.

Next, we examine the cumulative  roll yield across maturities. Denote by $T_1<T_2 <T_3<\dots$   the maturities of futures contracts.  We roll over at every $T_i$ by replacing the contract expiring at $T_i$ with a new contract that expires at $T_{i+1}$. Let $i(t):=\min\{i: T_{i-1}<t\leq T_{i}\}$, and $i(0)=1$. Then the roll yield   up to time  $t>T_1$ is
 \begin{align}
\mathcal{R}(0,t)&=(f^{T_{i(t)}}_t-f^{T_{i(t)}}_{T_{i(t)-1}})+\sum_{j=2}^{i(t)-1}(S_{T_j}-f^{T_j}_{T_{j-1}}) + (S_{T_1}-f^{T_1}_0) - (S_t-S_0)\notag\\
&= \underbrace{(f^{T_{i(t)}}_t - S_t) - (f^{T_1}_0-S_0)}_{\text{Basis Return}} + \underbrace{\sum_{j=1}^{i(t)-1}(S_{T_j}-f^{T_{j+1}}_{T_j})}_{\text{Cumulative Roll Adjustment}}.\label{rollyieldgeneral}
\end{align}
The cumulative roll adjustment is related to the term structure of futures contracts. If $T_{i}-T_{i-1}$ is constant, and the term structure only moves parallel, then the cumulative roll adjustment is simply the number of roll-over times a constant (difference between spot and near-month futures contract).

\subsection{OU and CIR Spot Models}
Suppose the spot price follows the OU or CIR model described in Section \ref{sect-OUCIRspot}. 
Inspecting \eqref{rollyieldgeneral}, we can write down the SDE for the roll yield  under the OU model:
\begin{align}\notag d\mathcal{R}(0,t) &= df^{T_{i(t)}}_t - dS_t\\
&= \left[e^{-\tilde{\mu}(T_{i(t)}-t)}\left(\mu(\theta - S_t) - \tilde{\mu}(\tilde{\theta} - S_t)\right) - \mu(\theta - S_t)\right]dt + \sigma\left(e^{-\tilde{\mu}(T_{i(t)}-t) }- 1\right) dB_t.\label{ryOU}
\end{align}
The roll yield SDE for under the CIR model has the same drift as \eqref{ryOU}. Furthermore,  the drift is positive iff
\begin{align*}
S_t > \frac{e^{-\tilde{\mu}(T_{i(t)}-t)}(\tilde{\mu}\tilde{\theta} - \mu \theta) + \mu \theta}{e^{-\tilde{\mu}(T_{i(t)}-t)}(\tilde{\mu}- \mu) + \mu}.
\end{align*}
In particular, if $\theta$ =  $\tilde{\theta}$, then the drift is positive iff $S_t > \theta.$ When $t = T_{i(t)},$ the drift is  $\tilde{\mu}(S_t - \tilde{\theta})$ and is positive iff $S_t > \tilde{\theta}.$ Furthermore,  the drift term can also be expressed as  
\begin{align*}
 \tilde{\mu}\left(f^{T_{i(t)}}_t - \tilde{\theta}\right) - \left(1 -  e^{-\tilde{\mu}(T_{i(t)}-t)}\right)\mu(\theta - S_t).
\end{align*}

On the other hand, we observe that
\begin{align*}
d\mathcal{R}(0,t)dS_t = \sigma^2\left(e^{-\tilde{\mu}(T_{i(t)}-t) }- 1\right) dt,
\end{align*}
under the OU case and 
\begin{align*}
d\mathcal{R}(0,t)dS_t = \sigma^2\left(e^{-\tilde{\mu}(T_{i(t)}-t) }- 1\right) S_t dt,
\end{align*}
under the CIR case. In other words, the instantaneous covariations betweeen roll yield and spot price under both OU and CIR models are negative for $t < T_{i(t)}$ regardless of the spot price level.

Consider a longer horizon with rolling at multiple maturities, the expected roll yield is
\begin{align*}
\E\{\mathcal{R}(0,t)\}&=\E\{f^{T_{i(t)}}_t - S_t\} - (f^{T_1}_0-S_0)+\sum_{j=1}^{i(t)-1}\E\{S_{T_j}-f^{T_{j+1}}_{T_j}\}\\
&= ((S_0-{\theta})e^{-{\mu}t}+{\theta}-\tilde{\theta})(e^{-\tilde{\mu}(T_{i(t)}-t)}-1)
-(S_0-\tilde{\theta})(e^{-\tilde{\mu} T_1}-1) \\
&\quad  +\sum_{j=1}^{i(t)-1}((S_0-{\theta})e^{-{\mu}T_j}+{\theta}-\tilde{\theta})(1-e^{-\tilde{\mu}(T_{j+1}-T_j)}).
\end{align*}

In summary, the expected roll yield depends not only on the risk-neutral parameters $\tilde{\mu}$ and $\tilde{\theta}$, but also their historical counterparts. It vanishes when $S_0 = \theta = \tilde{\theta}$. This is intuitive because if the current spot price is currently at the long-run mean, and the risk-neutral and historical measures coincide, then the spot and futures prices have  little tendency to deviate from the long-run mean. Also, notice that neither  the   futures price nor  the roll yield depends on the volatility parameter $\sigma$. This is true under the OU/CIR model, but not the exponential OU model, as we discuss next.

\subsection{Exponential OU Spot Model}
We now turn to  the exponential OU spot price model discussed in Section \ref{sect-XOUspot}. Recalling  the futures price in  \eqref{fTXOU},  the expected roll yield is given by
\begin{align}
&\E \{\mathcal{R}(0,t)\} = Y_1(t) + Y_2(t)  - (f^{T_1}_0-S_0), \label{xourollyield}
\end{align}
where 
\begin{align*}
Y_1(t) & = \E \{f^{T_{i(t)}}_t - S_t\} \\
&=\exp\bigg( e^{-\tilde{\mu}(T_{i(t)}-t)-{\mu}t}\ln(S_0)+ \bigg({\theta} -\frac{\sigma^2}{2{\mu}}\bigg)(1-e^{-{\mu}t})e^{-\tilde{\mu}(T_{i(t)}-t)}\\
&\quad +\frac{\sigma^2}{4{\mu}}e^{-2\tilde{\mu}(T_{i(t)}-t)}(1-e^{-2{\mu}t})+ (1-e^{-\tilde{\mu}(T_{i(t)}-t)})(\tilde{\theta}-\frac{\sigma^2}{2\tilde{\mu}}) \\
&\quad + \frac{\sigma^2}{4\tilde{\mu}}(1- e^{-2\tilde{\mu}(T_{i(t)}-t)} ) \bigg)\\
&\quad - \exp\left(e^{-{\mu}t}\ln(S_0)+(1-e^{-{\mu}t})({\theta}-\frac{\sigma^2}{2{\mu}})+\frac{\sigma^2}{4{\mu}}(1-e^{-{\mu}t}) \right),
\end{align*}
and 
\begin{align*}Y_2(t) &=\sum_{j=1}^{i(t)-1}\E \{S_{T_j}-f^{T_{j+1}}_{T_j}\} \\
&= \sum_{j=1}^{i(t)-1} \bigg( \exp\bigg(e^{-{\mu}T_j}\ln(S_0)+(1-e^{-{\mu}T_j})({\theta}-\frac{\sigma^2}{2{\mu}})+\frac{\sigma^2}{4{\mu}}(1-e^{-{\mu}T_j}) \bigg)\\
&\quad -\exp\bigg( e^{-\tilde{\mu}(T_{j+1}-T_j)-{\mu}T_j}\ln(S_0)+ \bigg({\theta} -\frac{\sigma^2}{2{\mu}}\bigg)(1-e^{-{\mu}T_j})e^{-\tilde{\mu}(T_{j+1}-T_j)} \\
&\quad +\frac{\sigma^2}{4{\mu}}e^{-2\tilde{\mu}(T_{j+1}-T_j)}(1-e^{-2{\mu}T_j})\\
&\quad  + (1-e^{-\tilde{\mu}(T_{j+1}-T_j)})(\tilde{\theta}-\frac{\sigma^2}{2\tilde{\mu}}) + \frac{\sigma^2}{4\tilde{\mu}}(1- e^{-2\tilde{\mu}(T_{j+1}-T_j)} ) \bigg)\bigg).
\end{align*}

 The explicit formula \eqref{xourollyield} for the  expected roll yield reveals the non-trivial dependence on the volatility parameter $\sigma$, as well as  the risk-neutral parameters $(\tilde{\mu},\tilde{\theta})$ and historical parameters $(\mu, \theta)$. It is useful for instantly  predicting the roll yield after calibrating the risk-neutral parameters   from the term structure of the futures prices, and estimating the historical parameters from past spot prices. 
 
Referring to \eqref{XOUP} and \eqref{xou_fut_sde}, the historical dynamics of the roll yield under an XOU spot model is given by
\begin{align*}
d\mathcal{R}(0,t) &= h(t,s)dt + \sigma\left(e^{-\tilde{\mu}(T_{i(t)}-t)}f^{T_{i(t)}}_t - S_t \right)dB_t,
\end{align*}
where  
\begin{align*}
h(t,s) &= \left(\ln s(\tilde{\mu} - \mu)  
+ (\mu\theta - \tilde{\mu}\tilde{\theta}) \right)\exp\bigg(e^{-\tilde{\mu}(T_{i(t)}-t)}\ln(s)  + (1-e^{-\tilde{\mu}(T_{i(t)}-t)})(\tilde{\theta}-\frac{\sigma^2}{2\tilde{\mu}}) \notag\\
&\qquad+ \frac{\sigma^2}{4\tilde{\mu}}(1- e^{-2\tilde{\mu}(T_{i(t)}-t)} )\bigg)e^{-\tilde{\mu}(T_{i(t)}-t)}- {\mu}({\theta}-\ln(s))s,
\end{align*} 
is the drift expressed in terms of the spot price $S_t.$ This reduces to 
\begin{align*}
h(T_{i(t)},s) = s\ln(s)(\tilde{\mu} - \mu + {\mu}{\theta}) - \tilde{\mu}\tilde{\theta}s, \quad \text{ if }~ t = T_{i(t)}.
\end{align*}
 Unlike the OU/CIR case, under an XOU spot model there is no explicit solution for the critical level of the spot price at which the drift changes sign.

As in the OU/CIR spot model, it is of interest to compute
 \begin{align*}
d\mathcal{R}(0,t)dS_t &= \sigma^2\left(e^{-\tilde{\mu}(T_{i(t)}-t)}f^{T_{i(t)}}_t - S_t \right)S_tdt,
\end{align*}
from which we see that the covariation between roll yield and spot price can be either positive or negative. In particular when the futures price is significantly higher than the spot price, i.e. when the market is in contango, the correlation tends to be positive.

\section{Optimal Timing to Trade Futures}\label{sect-futtrading}
In Section \ref{sect-futuresprices}, we have discussed the timing to liquidate a long  futures position, and the concept of rolling discussed in Section \ref{sect-rollyield} corresponds to holding the futures up to expiry. In this section, we further explore the timing options embedded in   futures, and develop the optimal trading strategies.

\subsection{Optimal Double Stopping Approach}
Let us consider the scenario in which  an investor has a long position in a futures contract with expiration date $T$. With a long position in the futures, the investor can hold it till maturity, but can also close the position  early by taking an opposite position at the prevailing  market price. At maturity, the two opposite positions cancel each other.   This motivates us to investigate the best time to close.

If the  investor selects to close the long position at time $\tau\le T$,  then she will receive the market value of the futures on the expiry date, denoted by  $f(\tau,S_\tau;T)$, minus the transaction cost $c \ge 0$.  To maximize the expected discounted value, evaluated under the investor's historical probability measure $\mathbb{P}$ with a constant subjective discount rate $r>0$,  the investor solves the optimal stopping problem 
 \begin{align*}
\VV(t, s) = \sup_{\tau \in\setT_{t,T}}{\E}_{t,s}\!\big\{e^{-r (\tau-t)}(f(\tau, S_{\tau};T) - c)\big\}, 
\end{align*}
where  $\setT_{t,T}$  is    the set of all stopping times, with respect to the filtration generated by $S$, taking values between $t$ and $\hat{T}$, where  $\hat{T} \in (0, T]$ is the trading deadline, which can equal but not exceed the futures' maturity.  Throughout this chapter, we continue to use the shorthand notation $\E_{t,s}\{\cdot\}\equiv\E\{\cdot|S_t=s\}$ to  indicate  the expectation taken under the historical probability measure $\mathbb{P}$.

The value function $\VV(t,s)$ represents the expected  liquidation
value associated with the  long futures position. Prior to taking the  long position in $f$, the investor, with zero position, can select the optimal timing to start the trade, or not to enter at all. This leads us to analyze the timing option  inherent in the trading
problem. Precisely, at time $t\le T$, the investor  faces the optimal entry timing problem  
\label{JJ1a}
\begin{align*}\JJ(t, s) =  \sup_{\nu \in \setT_{t, T}
}\E_{t,s}\!\left\{e^{-r (\nu-t)} (  \VV(\nu, S_{\nu})  - (f(\nu, S_{\nu};T) + \hat{c}))^{+}\right\},
\end{align*} 
where $\hat{c} \ge 0$ is the transaction cost, which may differ from $c$. In other words, the investor seeks
to maximize the expected difference between the value function
$\VV(\nu,S_\nu)$ associated with the long position and the prevailing   futures price $f(\nu, S_{\nu};T)$. The value function $\JJ(t,s)$
represents the maximum expected value of the  trading opportunity embedded in the futures.  We refer this ``long to open, short to close" strategy as the \emph{long-short} strategy\index{long-short strategy}.

Alternatively, an investor may well  choose to short a futures contract with the speculation that the futures price will fall, and then close it out later by establishing a long position.\footnote{By taking a short futures position, the  investor is required to  sell the underlying spot at maturity at a pre-specified price. In contrast to the short sale of a stock, a short futures does not involve share borrowing or re-purchasing.  } Given  an investor  who has  a unit  short position in the  futures contract, the objective is to minimize the expected discounted cost to close out this position at/before maturity.  The optimal timing strategy  is determined from 
\begin{align*}
\UU(t, s) =  \inf_{\tau \in \setT_{t,T}
}\E_{t,s}\!\left\{e^{-r (\tau-t)}(f(\tau, S_{\tau};T) + \hat{c})\right\}.
\end{align*} 
If the investor  begins with a zero position, then she can decide when to enter the market by solving
\begin{align*}
\KK(t, s) =  \sup_{\nu \in \setT_{t,T}
}\E_{t,s}\!\left\{e^{-r (\nu-t)} ((f(\nu, S_{\nu};T) - c) - \UU(\nu, S_{\nu}))^{+}\right\}.
\end{align*}
We call this  ``short to open, long  to close" strategy as the \emph{short-long} strategy\index{short-long strategy}. 

When an investor contemplates entering the market, she can either long or short first. Therefore, on top of the timing option, the investor has an additional choice between the long-short and short-long strategies. Hence, the investor solves the market entry timing problem:
 \begin{align}
\PP(t, s) =  \sup_{\vs \in \setT_{t,T}
}\E_{t,s}\!\left\{e^{-r (\vs-t)} {\max}\{\mathcal{A}(\vs, S_\vs), \mathcal{B}(\vs, S_\vs)\}\right\},\label{problemPP}
\end{align}
with two alternative rewards upon entry defined by
\begin{align*}
\mathcal{A}(\vs, S_{\vs}) &:= (\VV(\vs, S_{\vs})  - (f(\vs, S_{\vs};T) + \hat{c}))^{+} \quad \text{ (long-short)},\\
\mathcal{B}(\vs,S_{\vs}) &:= ((f(\vs, S_{\vs};T) - c) -\UU(\vs, S_{\vs}) )^{+} \quad \text {(short-long)}.\\
\end{align*}

\subsection{Variational Inequalities \& Optimal Trading Strategies}\label{sect-VI}
In order to solve for the optimal trading strategies, we study the  variational inequalities corresponding to the value functions $\JJ$, $\VV$, $\UU$, $\KK$ and $\PP$. To this end, we first define the operators:
 \begin{align}
 \L^{(1)}\{\cdot\}&:= -r \cdot + \frac{\partial\cdot}{\partial t} + \tilde{\mu}( \tilde{\theta} - s)\frac{\partial\cdot}{\partial s} + \frac{\sigma^2}{2}\frac{\partial^2 \cdot}{\partial s^2} \label{l1},\\
 \L^{(2)}\{\cdot\}&:= -r \cdot + \frac{\partial\cdot}{\partial t} + \tilde{\mu}( \tilde{\theta}  - s)\frac{\partial\cdot}{\partial s} + \frac{\sigma^2s}{2}\frac{\partial^2 \cdot}{\partial s^2}\label{l2},\\
 \L^{(3)}\{\cdot\}&:= -r \cdot + \frac{\partial\cdot}{\partial t} + \tilde{\mu}( \tilde{\theta}  - \ln s)\frac{\partial\cdot}{\partial s} + \frac{\sigma^2s^2	}{2}\frac{\partial^2 \cdot}{\partial s^2},\label{l3}
\end{align}
corresponding to,  respectively,  the OU, CIR, and XOU models. 

The optimal exit and entry problems  $\JJ$ and $\VV$ associated with the  \emph{long-short} strategy    are solved from  the following pair of variational inequalities:
\begin{align}
\textrm{max}\left\{\,\L^{(i)} \VV(t,s)\,, \,(f(t, s;T) -  c) - \VV(t,s)\,\right\} &= 0,\label{VIV}\\
\textrm{max}\left\{\,\L^{(i)} \JJ(t,s)\, ,\, (\VV(t,s)- (f(t, s;T)+\hat{c}))^{+} - \JJ(t,s) \, \right\} &=0, \label{VIJ}
\end{align}
for $(t,s) \in [0,T]\times  \mathbb{R}$, with $i\in\{1,2,3\}$ representing  the  OU, CIR, or XOU model respectively.\footnote{The spot price is positive, thus $s\in \mathbb{R}_+$,  under the CIR and XOU models.} 
Similarly, the reverse \textit{short-long} strategy can be determined by numerically solving the variational inequalities satisfied by $\UU$ and $\KK$:
\begin{align}
\textrm{min}\left\{\,\L^{(i)} \UU(t,s)\,,\,  (f(t, s;T) + \hat{c}) -\UU(t,s) \,\right\} &= 0,\label{VIU}\\
\textrm{max}\left\{\,\L^{(i)} \KK(t,s)\, , \,((f(t, s;T) - c) - \UU(t, s))^{+} - \KK(t,s) \, \right\} &=0. \label{VIK}
\end{align}
 As  $\VV$, $\JJ$, $\UU$, and $\KK$ are numerically solved, they become the input to the final problem represented by the value function $ \PP$. To determine the optimal timing to enter the futures market, we  solve the variational inequality 
 \begin{align}
\textrm{max}\left\{\,\L^{(i)} \PP(t,s)\,,\, \textrm{max}\{\mathcal{A}(t,s), \mathcal{B}(t,s)\} - \PP(t,s)\,  \right\} &=0. \label{VIP}
\end{align}

The optimal timing strategies are described by a series of boundaries representing the time-varying critical spot  price  at which the investor should establish a long/short futures position. In  the ``long to open, short to close" trading problem, where   the investor pre-commits to taking a long position first, the market entry timing is described by the ``$\JJ$'' boundary in Figure \ref{CIRVVJJ}. The subsequent timing to exit the market is represented by the ``$\VV$'' boundary in   Figure \ref{CIRVVJJ}. As we can see, the investor will long the futures when the spot  price is low, and short to close the position when the spot price is high, confirming the buy-low-sell-high intuition.

If the investor adopts the \textit{short-long} strategy, by which she will first short a futures and subsequently  close out with a long position, then  the optimal market entry  and exit timing strategies are  represented, respectively,  by the ``$\KK$'' and ``$\UU$'' boundaries in Figure \ref{CIRKKUU}. The investor will  enter the market by shorting a futures when the spot price is sufficiently high (at the ``$\KK$'' boundary), and  close it out when the spot price is low. Thus, the boundaries reflect a sell-high-buy-low strategy.

When there are no transaction costs (see Figure \ref{CIRVVJJ0} and \ref{CIRKKUU0}), the waiting region shrinks for both strategies. Practically, this means that the investor tends to enter and exit the market earlier, resulting in more rapid trades. This is intuitive as transaction costs discourage trades, especially near expiry. 

In the market entry problem represented by $\PP(t, s)$ in \eqref{problemPP}, the investor decides at what spot price to open a position. The corresponding timing strategy  is illustrated by two boundaries in Figure \ref{CIRoptimalbdy}. The boundary labeled as ``$\mathcal{P} = \AA$'' (resp. ``$\mathcal{P} = \BB$") indicates the critical spot price (as a function of time) at which the investor enters the market by taking a \emph{long} (resp. \textit{short}) futures position.  The area above the ``$\mathcal{P} = \BB$" boundary is the ``short-first" region, whereas the area below the ``$\mathcal{P} = \AA$'' boundary is the ``long-first" region. The area  between the two boundaries is the region where the investor should wait to enter. The ordering of the regions is intuitive -- the investor should long the futures when the spot price is currently low and short it when the spot price is high. As time approaches maturity, the value of entering the market diminishes. 
The investor will not start a long/short position unless the spot is very low/high close to maturity. Therefore, the waiting  region expands significantly near expiry.

\begin{figure}[h]
   \centering
\subfigure[]{ \includegraphics[trim=5   1  7  1,clip,width=2.3in]{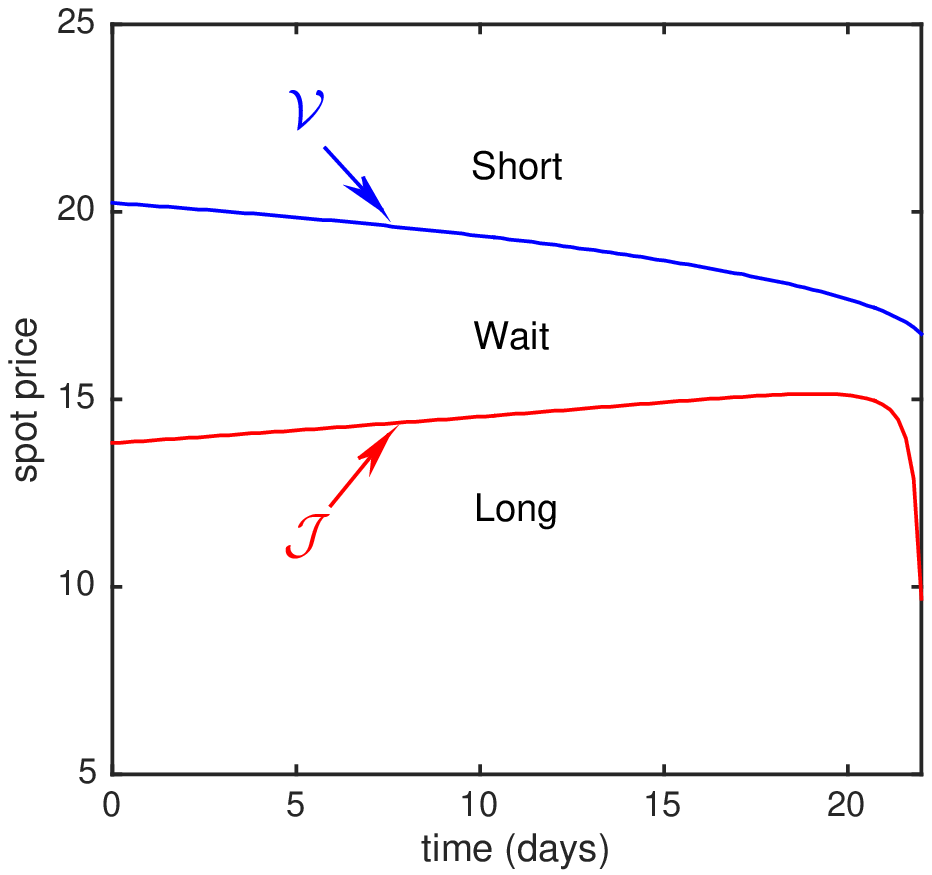}\label{CIRVVJJ}}
\subfigure[]{ \includegraphics[trim=5   1  7  1,clip,width=2.3in]{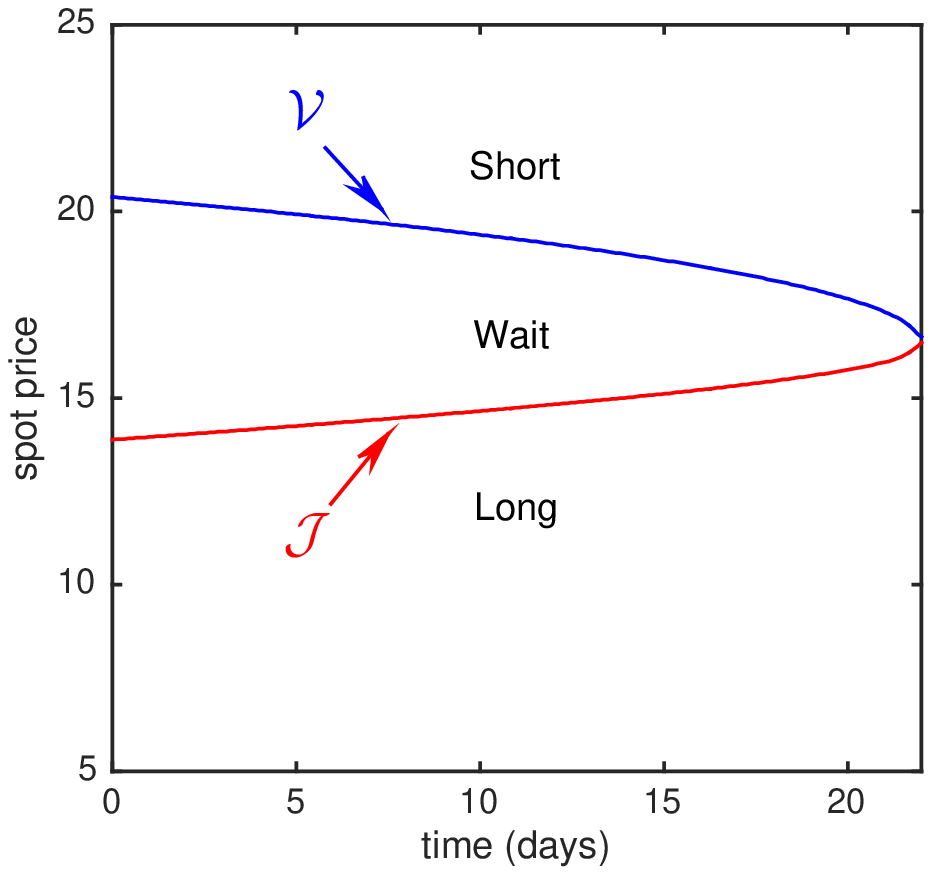}\label{CIRVVJJ0}}
\subfigure[]{ \includegraphics[trim=5  1 7  1,clip,width=2.3in]{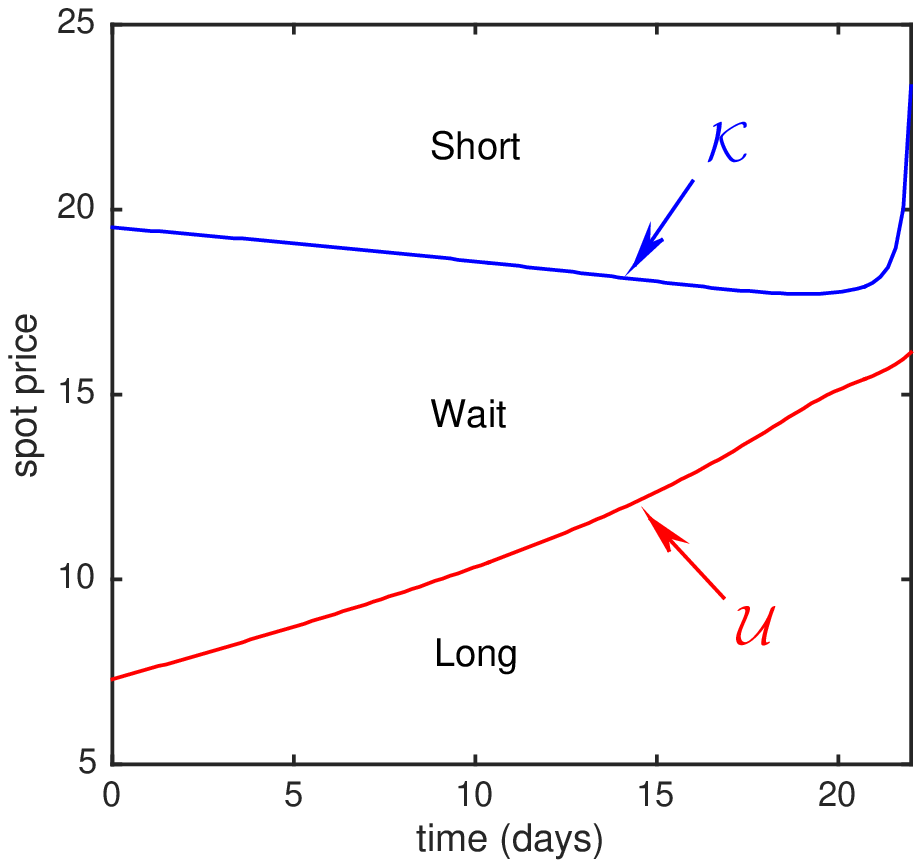}\label{CIRKKUU}}
\subfigure[]{ \includegraphics[trim=5  1 7  1,clip,width=2.3in]{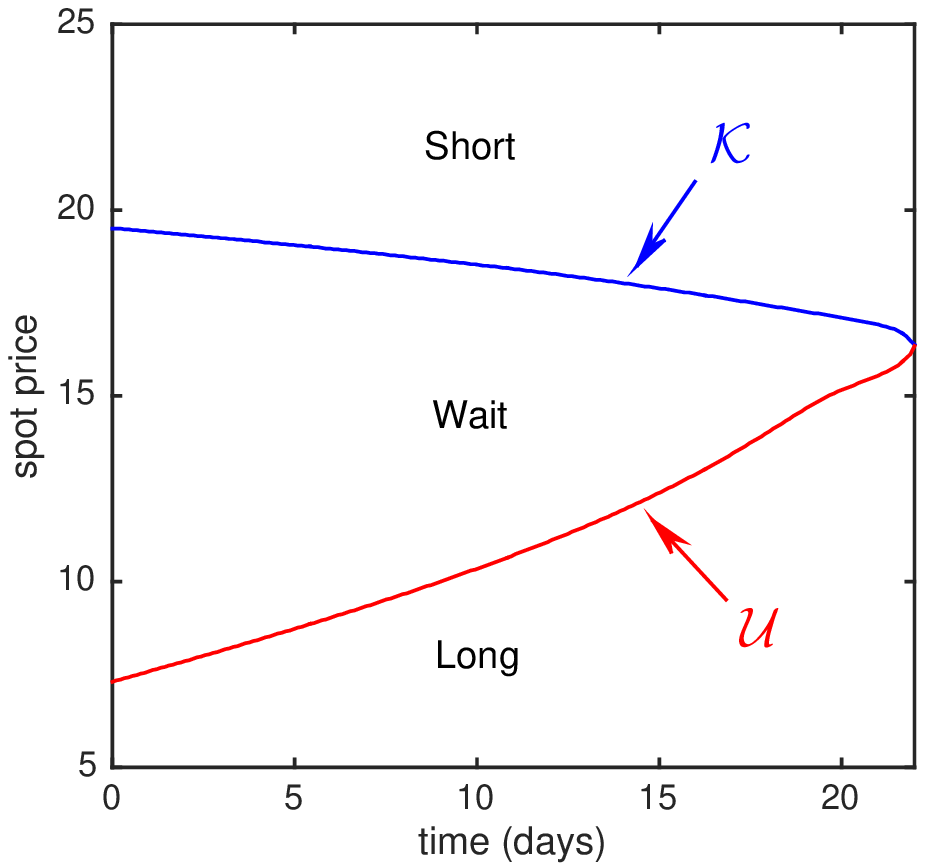}\label{CIRKKUU0}}
\subfigure[]{ \includegraphics[trim=5  1 7  1,clip,width=2.3in]{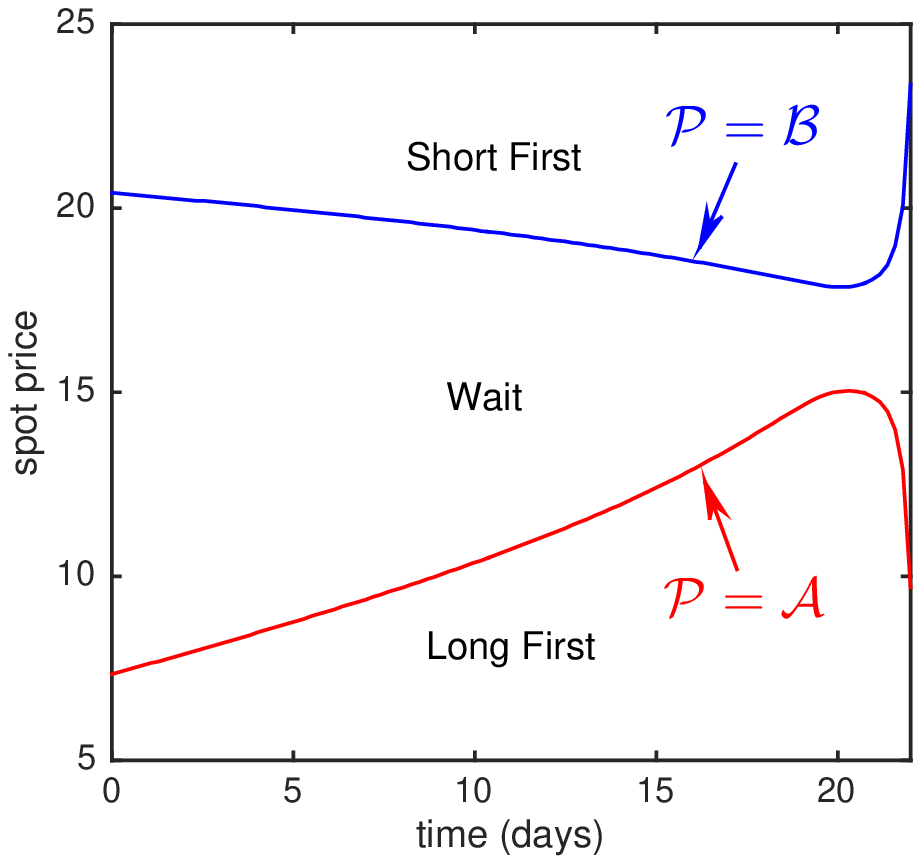}\label{CIRoptimalbdy}}
 \subfigure[]{ \includegraphics[trim=5  1 7  1,clip,width=2.3in]{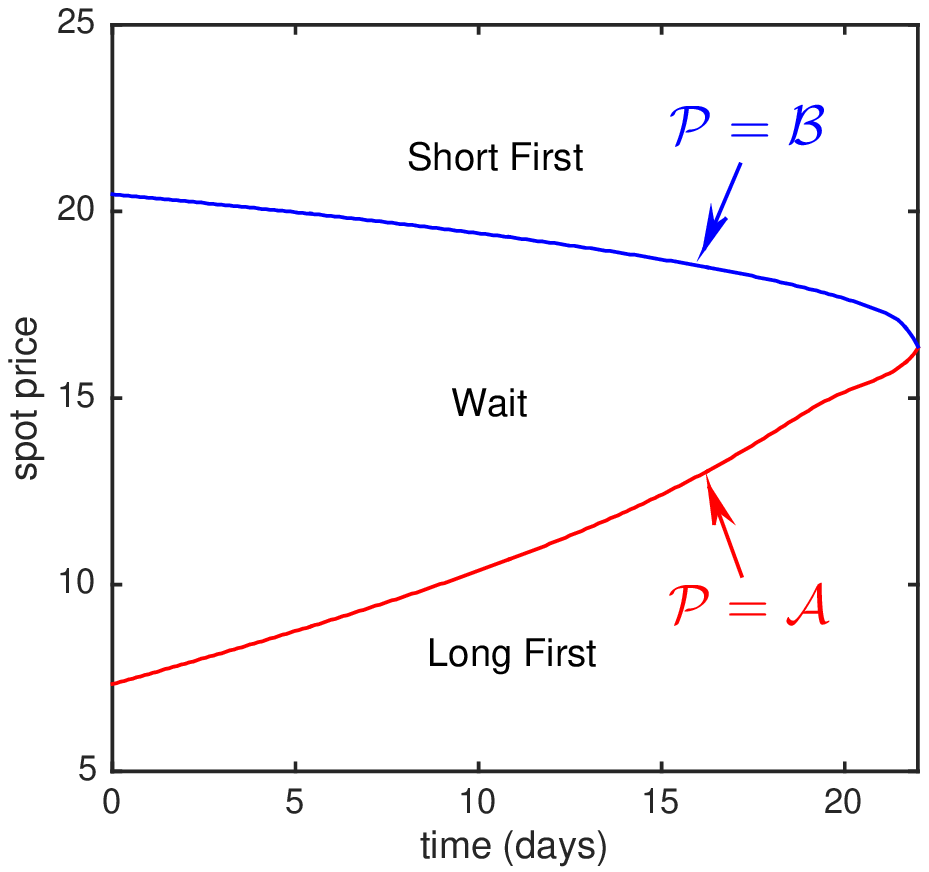}\label{CIRoptimalbdy0}}
\caption{Optimal long-short boundaries with/without transaction costs for futures trading under the CIR model in (a) and (b) respectively, optimal short-long boundaries with/without transaction costs in (c) and (d) respectively, and optimal boundaries with/without transaction costs in (e) and (f) respectively. Parameters: $\hat{T}=\frac{22}{252}$ , $T=\frac{66}{252}$, $ r=0.05, \sigma=5.33, \theta= 17.58,\tilde{\theta}=18.16,\mu=8.57,\tilde{\mu}=4.55, c=\hat{c}=0.005.$ }\label{CIRfigures}
\end{figure}

\clearpage

\begin{figure}[h]
\centering
\subfigure[]{ \includegraphics[width=4in]{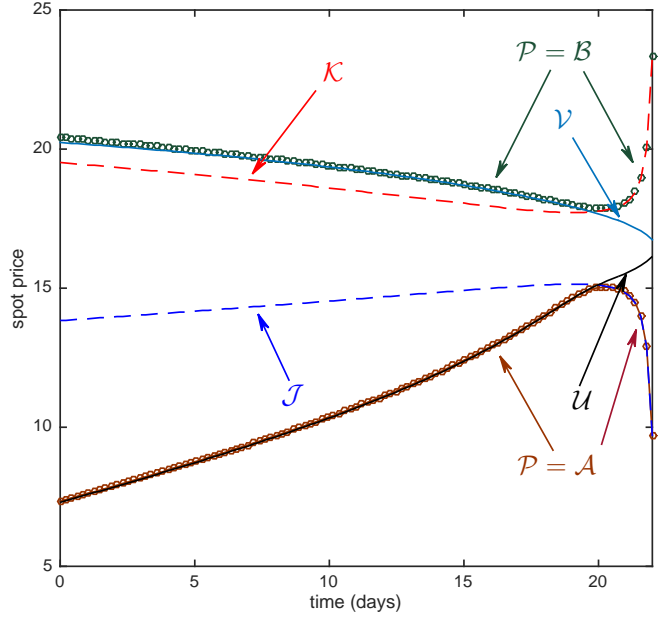}\label{CIRoptimalALL}}
\subfigure[]{ \includegraphics[width=4in]{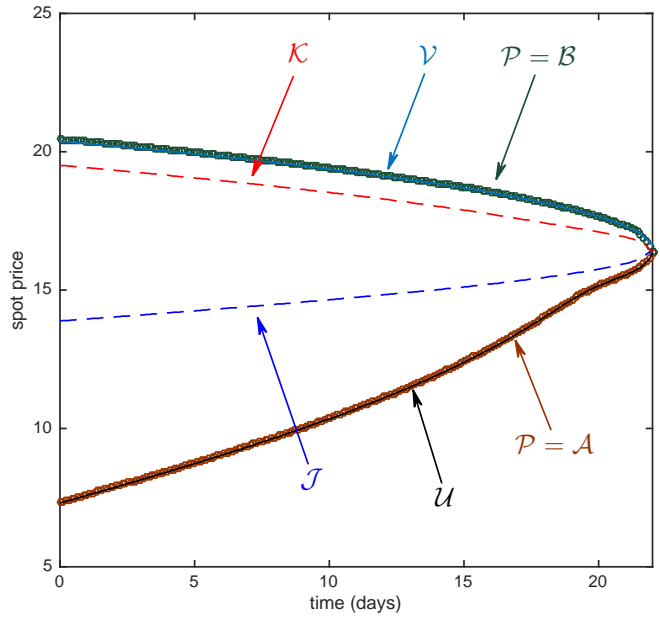}\label{CIRoptimalALL0}}
\caption{Optimal boundaries with and without transaction costs for futures trading under the CIR model in (a) and (b) respectively. Parameters: $\hat{T}=\frac{22}{252}$, $T=\frac{66}{252}$, $ r=0.05, \sigma=5.33, \theta= 17.58,\tilde{\theta}=18.16,\mu=8.57,\tilde{\mu}=4.55, c=\hat{c}=0.005. $}\label{CIRall}
\end{figure} 
 
 \clearpage

\begin{figure}[h]
   \centering
\subfigure[]{ \includegraphics[width=4in]{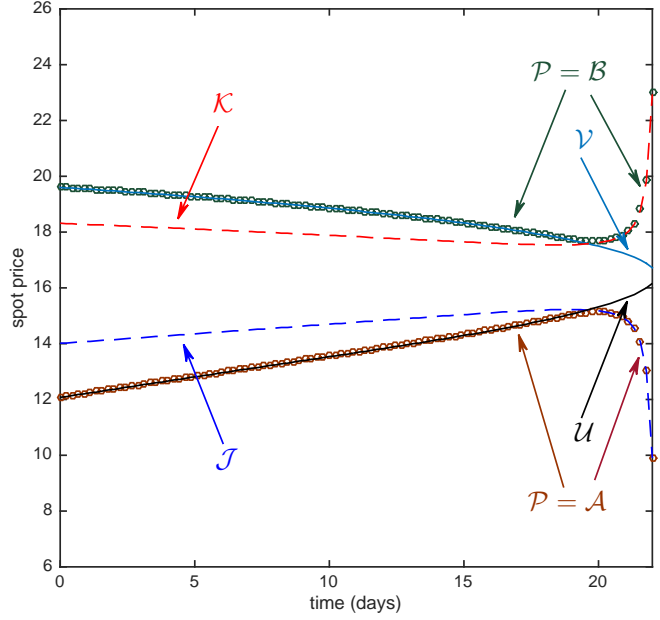}\label{OUoptimalALL}}
\subfigure[]{ \includegraphics[width=4in]{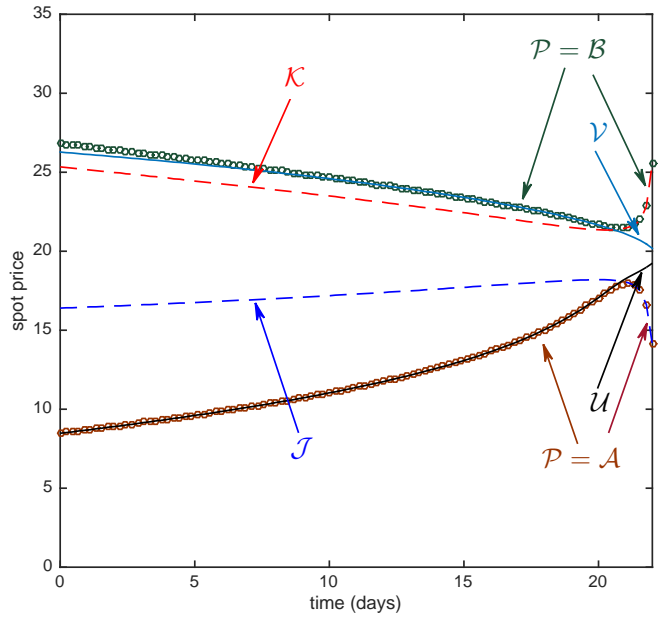}\label{XOUoptimalALL}}
\caption{Optimal boundaries with transaction costs for futures trading. (a) OU spot model with $\sigma=18.7, \theta= 17.58,\tilde{\theta}=18.16,\mu=8.57,\tilde{\mu}=4.55$. (b) XOU spot model with $\sigma=1.63, \theta= 3.03,\tilde{\theta}=3.06, \mu=8.57, \tilde{\mu}=4.08$. Common parameters: $\hat{T}=\frac{22}{252}$, $T=\frac{66}{252}$, $ c=\hat{c}=0.005.$}\label{OU_XOU_P}
\end{figure} 

  \clearpage

 \begin{figure}[th]
   \centering
     \includegraphics[width=4in]{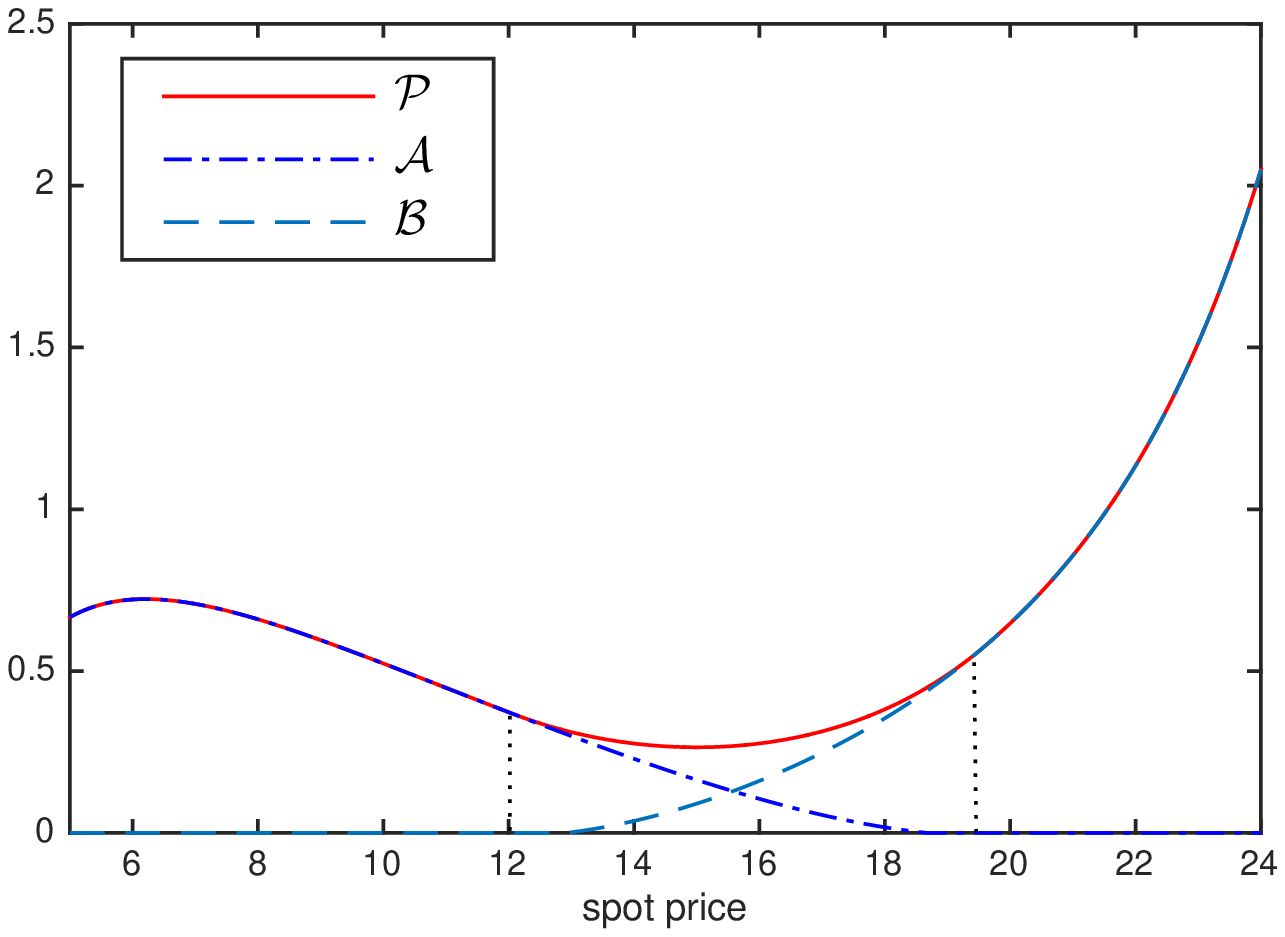}
     \caption{The value fuctions $\PP$, $\AA$, and $\BB$ plotted against the spot price at time $0$. The parameters are the same as those in Figure \ref{OU_XOU_P}.}\label{OUPAB}
\end{figure}

The investor's exit strategy depends on the initial  entry position. If the investor enters by taking a long position (at  the ``$\mathcal{P} = \AA$'' boundary), then the optimal exit timing to close her position is represented by the upper boundary with label  ``$\VV$'' in Figure \ref{CIRVVJJ}.  If the investor's initial position is short, then the optimal time to close by going long the futures is described by the lower boundary with label ``$\UU$'' in Figure \ref{CIRKKUU}.  

  Since the value function $\PP$ dominates both  $\JJ$ and $\KK$ due to the additional flexibility, it is not surprising that the  $``\PP = \AA$'' boundary  is lower than the ``$\JJ$'' boundary, and the  ``$\PP = \BB$'' boundary  is higher than the ``$\KK$'' boundary, as seen in Figure \ref{CIRoptimalALL}. This means that the embedded timing option to choose between the two strategies (``long to open, short to close"  or ``short to open, long to close") induces the  investor to delay market  entry to wait for better prices.   This phenomenon is also observed for both OU and XOU spot models in Figure \ref{OU_XOU_P}. Figure \ref{OUPAB} shows that the value function $\PP$ dominates $\BB$ and $\AA$ for all values of spot price. We can also see the regions where the ``$\PP = \AA$'' (when the spot price is low) and ``$\PP = \BB$'' (when the spot price is high).

  We see that Figure \ref{XOUoptimalALL} is similar to Figure \ref{CIRoptimalALL}, in both CIR and XOU cases, the difference between  ``$\UU$" boundary  and  the ``$\JJ$'' boundary is much larger than the difference between the  ``$\VV$'' boundary  and the ``$\KK$'' boundary. This means that the decision to choose  either   \textit{long-short} or \textit{short-long} has a larger impact on the optimal price level to long futures compared to the optimal level to short.  On the other hand, in Figure \ref{OUoptimalALL}, we observe a more symmetric relationship between the \textit{long-short} and \textit{short-long} optimal exercise boundaries. In particular, choosing one strategy or the other does not affect the optimal price levels as much as CIR and XOU cases.

 \section{Conclusion}
We have studied an optimal double stopping approach for trading futures under a number of mean-reverting spot models.  Our model yields trading decisions that are consistent with  the spot price dynamics and futures term structure. Accounting for the   timing options as well as  the option to choose between a long or short position, we find that it is optimal to  delay market entry, as compared to the case of committing to either  go long or short \emph{a priori}.

\clearpage
A natural direction for future research is to investigate the trading strategies under a multi-factor or time-varying mean-reverting spot price model. To this end, we include here some references that discuss the pricing aspect of futures under  such models, for example,  \cite{detempleVIX2000,LuZhuVixfutures2009,songVIX2012,futuresVIX} for VIX futures,    \cite{schwartz1997stochastic,CIRCommodity} for commodities, and  \cite{monoyiosfutures2002} for equity index futures.  It is also of practical interest to develop  similar  optimal multiple stopping approaches to trading commodities under mean-reverting spot models (\cite{LeungLiWang2014XOU, LeungLiWang2014CIR}),  and credit derivatives trading (\cite{LeungLiu2012}).
 
\section{Appendix}
\subsection{Numerical Implementation}
We apply a finite difference method to compute  the optimal boundaries in Figures \ref{CIRfigures}, \ref{CIRall} and  \ref{OU_XOU_P}. The   operators $\L^{(i)}$,  $i\in\{1,2,3\}$, defined in  \eqref{l1}-\eqref{l3} correspond to  the  OU, CIR, and  XOU models, respectively. To capture these models, we define the generic differential operator
\begin{align*}
 \L\{\cdot \}:= -r \cdot + \frac{\partial\cdot}{\partial t} + \varphi(s) \frac{\partial\cdot}{\partial s} + \frac{\sigma^2(s)}{2}\frac{\partial^2 \cdot}{\partial s^2},
\end{align*}
then the variational inequalities \eqref{VIV}, \eqref{VIJ}, \eqref{VIU}, \eqref{VIK} and \eqref{VIP} admit the same form as the following variational inequality problem:
\begin{align}
\begin{cases} 
\begin{split}
\L g(t,s) \leq 0, \enspace g(t,s) & \geq \xi (t,s),  \quad (t,s) \in [0,\hat{T}) \times \R_+,  \\ 
\\ (\L g(t,s)) (\xi (t,s) - g(t,s)) &= 0,  \quad (t,s) \in [0,\hat{T}) \times \R_+,\\
\\ g (\hat{T},s) &= \xi (\hat{T},s),  \quad s \in \R_+.
\end{split}	
\end{cases} \label{VIg}				
\end{align}
Here, $g(t,s)$   represents the value functions $\VV(t,s)$, $\JJ(t,s)$, $-\UU(t,s)$, $\KK(t,s)$, or $\PP(t,s)$. The function $\xi(t,s)$ represents $f(t,s;T) - c$, $(\VV(t,s) - (f(t,s;T) + \hat{c}))^+$, $-(f(t,s;T) + \hat{c})$, $(f(t,s;T) - c) - \UU(t,s))^+$, or $\max\{ \AA(t,s), \BB(t,s) \}$. The futures price $f(t,s;T)$, with $\hat{T} \le T$, is given by  \eqref{fTOU}, \eqref{fTCIR}, and \eqref{fTXOU} under  the OU, CIR, and XOU models, respectively.

We now consider the discretization of the partial differential equation $ \L g(t,s) =0$, over an uniform grid with discretizations in   time ($\delta t = \frac{\hat{T}}{N}$), and space    ($\delta s = \frac{S{\max}}{M}$).  We apply  the Crank-Nicolson method, which involves the finite difference equation:
\begin{align*}
-\alpha_i g_{i-1,j-1} + (1-\beta_i) g_{i,j-1} - \gamma_i g_{i+1,j-1}=\alpha_i g_{i-1,j} + (1+\beta_i) g_{i,j} + \gamma_i g_{i+1,j} ,
\end{align*}
where   
\begin{align}\notag
g_{i,j} &= g(j \delta t, i \delta s ), \quad  \xi_{i,j} = \xi(j \delta t, i \delta s ), \quad \varphi _i = \varphi (i \delta s), \quad   \sigma _i = \sigma (i \delta s). \\
\alpha_i &= \frac{\delta t}{4 \delta s}\big( \frac{\sigma ^2 _i }{\delta s} - \varphi _i \big),
\quad \beta_i = -\frac{\delta t}{2} \big(r + \frac{\sigma ^2 _i}{(\delta s)^2}\big), \quad 
\gamma_i = \frac{\delta t}{4 \delta s}\big( \frac{\sigma ^2 _i }{\delta s} + \varphi _i \big),\notag
\end{align}
for $i=1,2,...,M-1$ and $j=1,2,...,N-1$. The system to be solved  backward in time is 
\begin{align*}
\mathbf{M_1 g_{j-1}=r_j},
\end{align*}
where the right-hand side is 
\begin{align*}
\mathbf{r_j=M_2 g_{j}}+\alpha_1 \begin{bmatrix} g_{0,j-1}+g_{0,j} \\ 0 \\ \vdots \\0 \end{bmatrix} + \gamma_{M-1} \begin{bmatrix}  0 \\ \vdots \\0 \\ g_{M,j-1}+g_{M,j}, \end{bmatrix},
\end{align*}
and
\begin{align*}
\mathbf{M_1} &= \left[ \begin{array}{cccccc}
1- \beta _1 & -\gamma_1 & & & \\
-\alpha _2 & 1- \beta _2 & -\gamma_2 & & \\
& -\alpha _3 & 1- \beta _3 & -\gamma_3 & \\
& & \ddots & \ddots & \ddots \\
& & &- \alpha_{M-2} & 1- \beta _{M-2} & -\gamma_{M-2} \\
& & & & - \alpha_{M-1} & 1- \beta _{M-1} \end{array} \right],\\
\mathbf{M_2} &= \left[ \begin{array}{cccccc}
1+ \beta _1 & \gamma_1 & & & \\
\alpha _2 & 1+ \beta _2 & \gamma_2 & & \\
& \alpha _3 & 1+ \beta _3 & \gamma_3 & \\
& & \ddots & \ddots & \ddots \\
& & & \alpha_{M-2} & 1+ \beta _{M-2} & \gamma_{M-2} \\
& & & & \alpha_{M-1} & 1+ \beta _{M-1} \end{array} \right],\\
\mathbf{g_j} &=\begin{bmatrix} g_{1,j}, g_{2,j}, \hdots , g_{M-1,j} \end{bmatrix} ^T.
\end{align*}
 This  leads to a sequence of stationary complementarity problems. Hence, at each time step $j \in \left\{1, 2, \hdots, N-1\right\}$,  we need to solve 
\begin{align*}
\begin{cases} 
\begin{split}
\mathbf{M_1 g_{j-1}} & \geq \mathbf{r_j}, \\
\\ \mathbf{g_{j-1}} & \geq \boldsymbol{\xi _{j-1}},   \\ 
\\ (\mathbf{M_1 g_{j-1}} -\mathbf{r_j})^T (\boldsymbol{\xi _{j-1}} - \mathbf{g_{j-1}}) &= 0.  
\end{split}	
\end{cases} 			
\end{align*}
To solve the optimal problem, our algorithm enforces the constraint explicitly as follows
\begin{align}
g_{i,j-1}^{new}=\max \big\{g_{i,j-1}^{old},\xi_{i,j-1}\big\}.
\label{iterative}
\end{align}
The   projected SOR method is used to solve the linear system.\footnote{For a detailed discussion on the projected SOR method, we refer to \cite{wilmottbook1995}.}  At each time  $j$, we iteratively solve
\begin{align} 
\begin{split}
g_{1,j-1}^{(k+1)} &= \max \big\{\xi _{1,j-1} \,,\, g_{1,j-1}^{(k)} + \frac{\omega}{1-\beta_1} [r_{1,j}-(1-\beta_1) g_{1,j-1}^{(k)}+\gamma_1 g_{2,j-1}^{(k)}] \big\},\\
g_{2,j-1}^{(k+1)} &= \max \big\{\xi _{2,j-1} \,,\, g_{2,j-1}^{(k)} + \frac{\omega}{1-\beta_2} [r_{2,j}+\alpha_2 g_{1,j-1}^{(k+1)}-(1-\beta_2) g_{2,j-1}^{(k)}+\gamma_2 g_{3,j-1}^{(k)}] \big\},\\
\vdots\\
g_{M-1,j-1}^{(k+1)} &= \max \big\{\xi _{M-1,j-1} \,,\, g_{M-1,j-1}^{(k)} \\
&+ \frac{\omega}{1-\beta_{M-1}} [r_{M-1,j}+\alpha_{M-1} g_{M-2,j-1}^{(k+1)}-(1-\beta_{M-1}) g_{M-1,j-1}^{(k)}] \big\},
\end{split}
\label{PSOR}					
\end{align}
where $k$ is the iteration counter and $\omega$ is the overrelaxation parameter.
The iterative scheme starts from an initial point $\mathbf{g}_j ^{(0)}$ and proceeds until a convergence criterion is met, such as
$|| \mathbf{g}_{j-1} ^{(k+1)} - \mathbf{g}_{j-1} ^{(k)} || < \epsilon ,$ where $\epsilon$ is a tolerance parameter.  The  optimal boundary $S_f(t)$ can be identified by  locating the boundary that separates the  regions where  $g(t,s)=\xi(t,s)$, or  $g(t,s) \ge \xi(t,s)$.

\bibliographystyle{apa}    
 \bibliography{mybib2_10222015}    

\begin{thebibliography}{}

\bibitem[\protect\astroncite{Acworth}{2015}]{acworth2015}
Acworth, W. (2015).
\newblock 2014 {FIA} annual global futures and options volume: Gains in {N}orth
  {A}merica and {E}urope offset declines in {A}sia-{P}acific.
\newblock [Online; posted 09-March-2015].

\bibitem[\protect\astroncite{Bali and Demirtas}{2008}]{Balimeanreversion2008}
Bali, T.~G. and Demirtas, K.~O. (2008).
\newblock Testing mean reversion in financial market volatility: Evidence from
  {S\&P} 500 index futures.
\newblock {\em Journal of Futures Markets}, 28(1):1--33.

\bibitem[\protect\astroncite{Bessembinder et~al.}{1995}]{bessembinder1995}
Bessembinder, H., Coughenour, J.~F., Seguin, P.~J., and Smoller, M.~M. (1995).
\newblock Mean reversion in equilibrium asset prices: Evidence from the futures
  term structure.
\newblock {\em The Journal of Finance}, 50(1):361--375.

\bibitem[\protect\astroncite{Brennan and Schwartz}{1990}]{futuresBS}
Brennan, M.~J. and Schwartz, E.~S. (1990).
\newblock Arbitrage in stock index futures.
\newblock {\em Journal of Business}, 63(1):S7--S31.

\bibitem[\protect\astroncite{Cartea et~al.}{2015}]{HFTbook}
Cartea, A., Jaimungal, S., and Penalva, J. (2015).
\newblock {\em Algorithmic and High-Frequency Trading}.
\newblock Cambridge University Press, Cambridge, England.

\bibitem[\protect\astroncite{Casassus and Collin-Dufresne}{2005}]{Casassus2005}
Casassus, J. and Collin-Dufresne, P. (2005).
\newblock Stochastic convenience yield implied from commodity futures and
  interest rates.
\newblock {\em The Journal of Finance}, 60(5):2283--2331.

\bibitem[\protect\astroncite{Cox et~al.}{1981}]{CIR1981}
Cox, J.~C., Ingersoll, J., and Ross, S.~A. (1981).
\newblock The relation between forward prices and futures prices.
\newblock {\em Journal of Financial Economics}, 9(4):321--346.

\bibitem[\protect\astroncite{Dai et~al.}{2011}]{futuresDaiKwok}
Dai, M., Zhong, Y., and Kwok, Y.~K. (2011).
\newblock Optimal arbitrage strategies on stock index futures under position
  limits.
\newblock {\em Journal of Futures Markets}, 31(4):394--406.

\bibitem[\protect\astroncite{Detemple and Osakwe}{2000}]{detempleVIX2000}
Detemple, J. and Osakwe, C. (2000).
\newblock The valuation of volatility options.
\newblock {\em European Finance Review}, 4(1):21--50.

\bibitem[\protect\astroncite{Elton et~al.}{2009}]{mpt2009}
Elton, E.~J., Gruber, M.~J., Brown, S.~J., and Goetzmann, W.~N. (2009).
\newblock {\em Modern Portfolio Theory and Investment Analysis}.
\newblock Wiley, 8th edition.

\bibitem[\protect\astroncite{Geman}{2007}]{gemanoilmeanreversion2007}
Geman, H. (2007).
\newblock Mean reversion versus random walk in oil and natural gas prices.
\newblock In Fu, M.~C., Jarrow, R.~A., Yen, J.-Y.~J., and Elliot, R.~J.,
  editors, {\em Advances in Mathematical Finance}, Applied and Numerical
  Harmonic Analysis, pages 219--228. BirkhŠuser Boston.

\bibitem[\protect\astroncite{Gorton et~al.}{2013}]{gorton2013fundamentals}
Gorton, G.~B., Hayashi, F., and Rouwenhorst, K.~G. (2013).
\newblock The fundamentals of commodity futures returns.
\newblock {\em Review of Finance}, 17(1):35--105.

\bibitem[\protect\astroncite{Gr\"{u}bichler and
  Longstaff}{1996}]{Grunbichler1996985}
Gr\"{u}bichler, A. and Longstaff, F. (1996).
\newblock Valuing futures and options on volatility.
\newblock {\em Journal of Banking and Finance}, 20(6):985--1001.

\bibitem[\protect\astroncite{Irwin
  et~al.}{1996}]{scottfuturesmeanreversion1996}
Irwin, S.~H., Zulauf, C.~R., and Jackson, T.~E. (1996).
\newblock Monte {C}arlo analysis of mean reversion in commodity futures prices.
\newblock {\em American Journal of Agricultural Economics}, 78(2):387--399.

\bibitem[\protect\astroncite{Leung and Li}{2015}]{LeungLi2014OU}
Leung, T. and Li, X. (2015).
\newblock Optimal mean reversion trading with transaction costs and stop-loss
  exit.
\newblock {\em International Journal of Theoretical \& Applied Finance},
  18(3):15500.

\bibitem[\protect\astroncite{Leung et~al.}{2014}]{LeungLiWang2014CIR}
Leung, T., Li, X., and Wang, Z. (2014).
\newblock Optimal starting--stopping and switching of a {CIR} process with
  fixed costs.
\newblock {\em Risk and Decision Analysis}, 5(2):149--161.

\bibitem[\protect\astroncite{Leung et~al.}{2015}]{LeungLiWang2014XOU}
Leung, T., Li, X., and Wang, Z. (2015).
\newblock Optimal multiple trading times under the exponential {OU} model with
  transaction costs.
\newblock {\em Stochastic Models}, 31(4).

\bibitem[\protect\astroncite{Leung and Liu}{2012}]{LeungLiu2012}
Leung, T. and Liu, P. (2012).
\newblock Risk premia and optimal liquidation of credit derivatives.
\newblock {\em International Journal of Theoretical \& Applied Finance},
  15(8):1250059.

\bibitem[\protect\astroncite{Leung and Shirai}{2015}]{LeungShirai}
Leung, T. and Shirai, Y. (2015).
\newblock Optimal derivative liquidation timing under path-dependent risk
  penalties.
\newblock {\em Journal of Financial Engineering}, 2(1):1550004.

\bibitem[\protect\astroncite{Lu and Zhu}{2009}]{LuZhuVixfutures2009}
Lu, Z. and Zhu, Y. (2009).
\newblock Volatility components: The term structure dynamics of {VIX} futures.
\newblock {\em Journal of Futures Markets}, 30(3):230--256.

\bibitem[\protect\astroncite{Menc{\'i}a and Sentana}{2013}]{futuresVIX}
Menc{\'i}a, J. and Sentana, E. (2013).
\newblock Valuation of {VIX} derivatives.
\newblock {\em Journal of Financial Economics}, 108(2):367--391.

\bibitem[\protect\astroncite{Monoyios and Sarno}{2002}]{monoyiosfutures2002}
Monoyios, M. and Sarno, L. (2002).
\newblock Mean reversion in stock index futures markets: a nonlinear analysis.
\newblock {\em The Journal of Futures Markets}, 22(4):285--314.

\bibitem[\protect\astroncite{Moskowitz et~al.}{2012}]{moskowitz2012time}
Moskowitz, T.~J., Ooi, Y.~H., and Pedersen, L.~H. (2012).
\newblock Time series momentum.
\newblock {\em Journal of Financial Economics}, 104(2):228--250.

\bibitem[\protect\astroncite{Ribeiro and Hodges}{2004}]{CIRCommodity}
Ribeiro, D.~R. and Hodges, S.~D. (2004).
\newblock A two-factor model for commodity prices and futures valuation.
\newblock {EFMA} 2004 {B}asel Meetings Paper.

\bibitem[\protect\astroncite{Schwartz}{1997}]{schwartz1997stochastic}
Schwartz, E. (1997).
\newblock The stochastic behavior of commodity prices: Implications for
  valuation and hedging.
\newblock {\em The Journal of Finance}, 52(3):923--973.

\bibitem[\protect\astroncite{Wang and Daigler}{2011}]{wangVIX2011}
Wang, Z. and Daigler, R.~T. (2011).
\newblock The performance of {VIX} option pricing models: Empirical evidence
  beyond simulation.
\newblock {\em Journal of Futures Markets}, 31(3):251--281.

\bibitem[\protect\astroncite{Wilmott et~al.}{1995}]{wilmottbook1995}
Wilmott, P., Howison, S., and Dewynne, J. (1995).
\newblock {\em The Mathematics of Financial Derivatives: A Student
  Introduction}.
\newblock Cambridge University Press, 1st edition.

\bibitem[\protect\astroncite{Zhang and Zhu}{2006}]{futures_zhang}
Zhang, J.~E. and Zhu, Y. (2006).
\newblock {VIX} futures.
\newblock {\em Journal of Futures Markets}, 26(6):521--531.

\bibitem[\protect\astroncite{Zhu and Lian}{2012}]{songVIX2012}
Zhu, S.-P. and Lian, G.-H. (2012).
\newblock An analytical formula for {VIX} futures and its applications.
\newblock {\em Journal of Futures Markets}, 32(2):166--190.

\end{thebibliography}
 
\end{document}